\documentclass[lettersize,journal]{IEEEtran}
\usepackage{amsmath,amsfonts,amssymb,amsthm,mathtools}
\usepackage{algorithm}
\usepackage{array}
\usepackage{textcomp}
\usepackage{stfloats}
\usepackage{verbatim}
\usepackage{graphicx}
\usepackage{cite}
\usepackage{algorithmicx,algpseudocode}
\usepackage{graphicx}
\usepackage{textcomp}
\usepackage{xcolor}
\usepackage[labelformat=simple]{subcaption}
\usepackage{multirow}
\usepackage{makecell}
\usepackage{diagbox}
\usepackage{pdfpages}
\usepackage{bm}
\usepackage{listings}%
\usepackage{enumitem} \setlist[itemize]{leftmargin=*}
\usepackage{booktabs}
\usepackage{setspace,hyperref}
\usepackage{color}
\usepackage{etoolbox}
\usepackage{circuitikz}
\usepackage{tikz}
\usepackage[normalem]{ulem}
\usepackage{cuted}
\usepackage{acro}
\useunder{\uline}{\ul}{}
\usetikzlibrary{positioning}
\usetikzlibrary{shapes.symbols}
\usetikzlibrary{angles,quotes}
\usepackage{pgfkeys}

\newtheorem{thm}{Theorem}

\newtheorem{coro}{Corollary}

\newtheorem{lemma}{Lemma}

\DeclareMathOperator{\diag}{diag}
\DeclareMathOperator{\supp}{supp}
\DeclareMathOperator{\trace}{tr}

\DeclareMathOperator{\vect}{vec}

\DeclareAcronym{BEM}{
  short = BEM ,
  long  = basis expansion model ,
  short-plural = s , 
  long-plural  = s 
}
\DeclareAcronym{OMP}{
  short = OMP ,
  long  = orthogonal matching pursuit ,
  short-plural = s , 
  long-plural  = s 
}
\DeclareAcronym{SBL}{
  short = SBL ,
  long  = sparse Bayesian learning ,
  short-plural = s , 
  long-plural  = s 
}

\begin{document}

\title{Efficient Off-Grid Bayesian Parameter Estimation for Kronecker-Structured Signals}

\author{Yanbin He and Geethu Joseph
        % <-this % stops a space ~\IEEEmembership{Staff,~IEEE,}
\thanks{
The material in this paper was presented
in part in~\cite{he2023kronecker}.

The authors are with the Signal Processing Systems group, Electrical Engineering, Mathematics, and Computer Science faculty, at the Delft University of Technology, The Netherlands. Emails:$\{\text{y.he-1, g.joseph}\}\text{@tudelft.nl}$.}% <-this % stops a space
}



\maketitle

\begin{abstract}
%The multidimensional signals refer to the signal jointly from multiple dimensions, e.g., spatial, temporal, and frequency. Such signals inherit information from multiple dimensions and can be used to estimate unknown parameters in these dimensions, and generally involve the Kronecker product. 
This work studies the problem of jointly estimating unknown parameters from Kronecker-structured multidimensional signals, which arises in applications like intelligent reflecting surface (IRS)-aided channel estimation. Exploiting the Kronecker structure, we decompose the estimation problem into smaller, independent subproblems across each dimension. Each subproblem is posed as a sparse recovery problem using basis expansion and solved using a novel off-grid sparse Bayesian learning (SBL)-based algorithm. Additionally, we derive probabilistic error bounds for the decomposition, quantify its denoising effect, and provide convergence analysis for off-grid SBL. Our simulations show that applying the algorithm to IRS-aided channel estimation improves accuracy and runtime compared to state-of-the-art methods through the low-complexity and denoising benefits of the decomposition step and the high-resolution estimation capabilities of off-grid SBL.
%Our simulation results evaluate the decomposition step and off-grid SBL separately, highlighting the computational efficiency and denoising benefits of the decomposition step and the high-resolution estimation capabilities of off-grid SBL. Finally, applying our algorithm to IRS-aided channel estimation demonstrates improved accuracy and reduced runtime compared to the state-of-the-art methods, driven by the combined effects of decomposition and grid optimization. 
\end{abstract}


\section{Introduction}\label{sec.intro}

\IEEEPARstart{M}{ultidimensional} signals arise in several engineering applications such as image processing \cite{caiafa2012block,caiafa2013computing,caiafa2013multidimensional} and wireless communications \cite{zhou2017low,wang2024tensor,he2023bayesian}. In these contexts, the data is represented as a function of different dimensions, each conveying a specific physical quantity. For example, in the uplink narrowband intelligent reflecting surface (IRS)-aided system, the received signal at the base station (BS) from the mobile station (MS) is a function of angle-of-departure (AoD) at MS, the difference of angle-of-arrival (AoA) and AoD at the IRS, and AoA at BS~\cite{he2022structure}. Considering the angular domain of each array as a separate dimension, this signal is multidimensional \cite{he2022structure,he2023bayesian}. This structure is captured by the Kronecker product \cite{ortiz2019sparse}, leading to the fundamental model,
\begin{equation}\label{eq.kro_basic}
    \bar{\bm y} = \bm y_1 \otimes \bm y_2 \otimes \cdots \otimes \bm y_I +\bar{\bm n}= \otimes_{i=1}^I \bm y_i +\bar{\bm n},
\end{equation}
where $\otimes_{i=1}^I \bm y_i\in \mathbb{C}^{\bar{M}}$ is an $I$-dimensional signal, $\bm y_i\in \mathbb{C}^{M_i}$ represents the signal in each dimension, $\bar{\bm n}$ is the noise, $\bar{M}=\prod_{i=1}^IM_i$, and $\otimes$ is the Kronecker product. Each $\bm y_i$ encapsulates the signal in the corresponding dimension (e.g., AoA and AoD) and is expressed as a weighted sum of nonlinear parametric functions,
\begin{equation}\label{eq.dimen_linear}
  \bm y_i = \sum_{s=1}^{S_i}\bm h_i(\bar{\psi}_{i,s}) \bar{x}_{i,s},
\end{equation}
with parameters $\bar{\psi}_{i,s}$ and weights $\bar{x}_{i,s}$, for $s=1,2,\ldots,S_i$, where $\bm h_i(\cdot)\in\mathbb{C}^{M_i}$ is the nonlinear function. In the IRS-aided system example, the parameter $\bar{\psi}_{i,s}$ can be AoAs or AoDs, the nonlinear function $\bm h_i$ is related to the steering vector, and $\bar{x}_{i,s}$ represents the path gain corresponding to each AoA or AoD. Thus, the channel estimation problem reduces to estimating all $\bar{\psi}_{i,s}$'s and $\bar{x}_{i,s}$'s from the received signal $\bar{\bm y}$, where $S_i$'s are also unknown. Hence, this paper focuses on the general problem of estimating the parameters and weights $\{\bar{\psi}_{i,s},\bar{x}_{i,s}\}_{s=1}^{S_i}$ from measurements $\bar{\bm y}$ and the function $\bm h_i(\cdot)$ for all $I$ dimensions. \IEEEpubidadjcol

A popular approach for parameter estimation from~\eqref{eq.kro_basic} and~\eqref{eq.dimen_linear} is multidimensional \ac{BEM}~ \cite{caiafa2012block,caiafa2013computing,caiafa2013multidimensional,chang2021sparse,xu2022sparse,he2022structure}. It evaluates the nonlinear function over pre-sampled \emph{grids} of unknown parameters in each dimension to express the multidimensional signal as the product of a known overcomplete Kronecker-structured dictionary of the basis functions and an unknown sparse coefficient vector as
\begin{equation}\label{eq.problem_basic}
    \bar{\bm y} = \left(\otimes_{i=1}^I \bm H_i\right) \bm x+\bar{\bm n}.
\end{equation}
Here, $\bar{\bm y} \in \mathbb{C}^{\bar{M}}$ is the measurement, $\bm H_i \in \mathbb{C}^{M_i \times N_i}$ is the overcomplete basis for $\bm h_i(\cdot)$, and $\otimes_{i=1}^I \bm H_i \in \mathbb{C}^{\bar{M} \times \bar{N}}$ is the overall dictionary with $\bar{M} < \prod_{i=1}^I N_i = \bar{N}$. Also, $\bm x \in \mathbb{C}^{\bar{N}}$ is the unknown sparse vector, and $\bar{\bm n}$ is the measurement noise. The model in \eqref{eq.dimen_linear} leads to a Kronecker-structure in $\bm x$ given by
\begin{equation}\label{eq.vect_kro}
    \bm x = \otimes_{i=1}^I \bm x_i,
\end{equation}
where $\bm x_i \in \mathbb{C}^{N_i}$ is the weight in the $i$th dimension. The overcomplete dictionary makes $\bm x_i$ sparse, with only $S_i$ nonzero entries corresponding to the true parameters. Thus, estimating parameters and coefficients is a sparse recovery problem. 

Solving \eqref{eq.problem_basic} for a sparse vector $\bm x$ with Kronecker-structured support has been discussed in \cite{chang2021sparse,xu2022sparse,he2023bayesian,caiafa2013computing}. A greedy method, Kronecker-\ac{OMP}, generalizes the traditional \ac{OMP} to multidimensional \ac{BEM}~\cite{caiafa2013computing}. It has low complexity but requires hand-tuning of a sensitive stopping threshold~\cite{he2023bayesian}. Another approach with improved accuracy and no hand-tuning relies on \ac{SBL}. It adopts a fictitious Gaussian prior on the sparse vector with a Kronecker-structured covariance matrix to enforce a Kronecker-structured support~\cite{chang2021sparse,xu2022sparse,he2023bayesian}. %However, the derivations in \cite{chang2021sparse,xu2022sparse} rely on approximations, leading to compromised performance. \cite{he2023bayesian} further devises Kronecker-\ac{SBL} (KroSBL), which not only effectively incorporates the prior knowledge of the Kronecker-structured support but also has theoretically guaranteed performance. 
Although Kronecker-\ac{SBL} (KroSBL) can be readily applied to our problem, it has two main drawbacks, as elaborated below.

First, KroSBL does not fully exploit prior knowledge~\eqref{eq.vect_kro}. 
%KroSBL imposes a fictitious Gaussian prior on $\bm x$ with Kronecker-structured covariance matrix to mimic the Kronecker structure in~\eqref{eq.vect_kro}. However, since 
While KroSBL employs a Kronecker-structured covariance matrix, 
the variance only determines whether an entry is nonzero; in effect, KroSBL only exploits the Kronecker structure of the support vector. Consequently, KroSBL directly estimates the high-dimensional vector $\bm x$. Although the state-of-the-art KroSBL algorithm employs some complexity reduction techniques~\cite{he2023bayesian}, it still faces high overall complexity due to the high dimensionality of the Kronecker product. So, we seek a method that can fully exploit the Kronecker structure while significantly reducing complexity compared to KroSBL. 

Second, KroSBL relies on multidimensional \ac{BEM}, formulating the dictionary using predefined grids. However, the true parameters may not fall on these grids, causing a \emph{grid mismatch} issue \cite{zhu2011sparsity}, which can degrade the estimation performance and cannot be resolved with finer grids \cite{tang2013compressed}. Addressing mismatch is challenging due to the \emph{non-linearity} of the \ac{SBL} cost function. 
%To handle non-linearity, linearization and marginal likelihood maximization are commonly adopted. Linearization approximates the non-linear function around the grids, often with Taylor expansion \cite{zhu2011sparsity,yang2012off,you2020parametric}. Off-grid optimization is then achieved by seeking the first-order coefficient of the grids rather than directly optimizing the grids. However, such an approximation is valid only in the vicinity of grid points \cite{ibrahim2014estimation}. Finer grids can alleviate but cannot fully eliminate this issue~\cite{mao2021marginal,dai2018fdd}. 
Linearization methods, such as Taylor expansion~\cite{zhu2011sparsity,yang2012off,you2020parametric}, approximate the non-linear function near grids and optimize first-order coefficients, but their validity is limited to grid vicinities~\cite{mao2021marginal,dai2018fdd}. Alternatively, some works use marginal likelihood maximization to sequentially optimize each grid's contribution to the maximum likelihood (ML) cost ~\cite{faul2001analysis}, to either refine~\cite{liu2012efficient,mao2021marginal} or determine the grid points~\cite{pote2023light,ament2021sparse}. 
%The second approach, marginal likelihood maximization, allows isolation of each grid's contribution to the maximum likelihood (ML) cost function, enabling alternating optimization for each grid rather than simultaneous optimization of all grids~\cite{faul2001analysis}. This method can either refine the grid based on a rough estimate \cite{liu2012efficient,mao2021marginal} or directly determine which grid points to include or exclude in a forward or backward regression manner \cite{pote2023light,ament2021sparse}. 
However, this approach exhibits a greedy nature~\cite{ament2021sparse} and can lead to performance degradation as the number of unknowns increases \cite{lin2022covariance}. The drawbacks of existing multidimensional \ac{BEM} and grid-less SBL algorithms motivate novel approaches to solving our parameter estimation problem. 

%This work focuses on estimating the parameters $\{\bar{\psi}_{i,s}\}_{s=1}^{S_i}$ and weights $\{\bar{x}_{i,s}\}_{s=1}^{S_i}$ for all $I$ dimensions using $\bar{\bm y}$ in~\eqref{eq.kro_basic} and~\eqref{eq.dimen_linear}. 
We aim to develop a method for estimating the parameters and weights $\{\bar{\psi}_{i,s},\bar{x}_{i,s}\}_{s=1}^{S_i}$ for all $I$ dimensions using $\bar{\bm y}$ in~\eqref{eq.kro_basic} and~\eqref{eq.dimen_linear} with three key features: $(i)$ fully utilizing the Kronecker structure in~\eqref{eq.kro_basic}; $(ii)$ overcoming the grid mismatch of linearization and marginal likelihood optimization; and $(iii)$ achieving lower complexity compared to KroSBL. Our algorithm follows the \ac{BEM} paradigm using \ac{SBL} and enjoys the theoretical guarantees. Our main contributions are as follows:
\begin{itemize}
    \item \emph{Decomposition-based Algorithm}: We decompose the measurement $\bar{\bm y}$ into multiple low-dimensional measurements, fully utilizing the prior information of the Kronecker structure in Sec.~\ref{sec.decomposition}. It transforms the joint multidimensional unknown parameters estimation into multiple separate subproblems in each dimension, leading to reduced complexity.

    \item \emph{Off-grid Algorithm}: We use \ac{BEM} for parameters estimation in each dimension and cast it into a sparse vector recovery problem solved using the expectation-maximization (EM)-based \ac{SBL} in Sec.~\ref{sec.mbem}. We further incorporate a grid optimization step in the EM iterations to avoid grid mismatch.

    \item \emph{Algorithm Analyses and Extensions}: We study the decomposition step and the iterative grid optimization in Sec.~\ref{sec.theoretical}. We theoretically quantify the error bound of the decomposition step in the presence of noise and the denoising effect which we attribute to the better estimation performance. We discuss the convergence property of our algorithm and its potential extensions to related measurement structures arising from other scenarios in Sec.~\ref{sec.extension}.

    \item \emph{Application}: In Sec.~\ref{sec.channelmodel}, we analyze the signal model of a prototypical IRS-aided wireless communication system and explain the implementation of our algorithm for uplink cascaded IRS channel estimation.

    \item \emph{Numerical Results}: 
We evaluate our schemes in three scenarios in Sec.~\ref{sec.simu}. The first scenario highlights the computational efficiency and denoising benefits of the decomposition method. The second scenario demonstrates the high-resolution estimation capabilities of off-grid SBL. The third scenario focuses on IRS channel estimation, showcasing improved accuracy and reduced runtime, driven by the combined effects of decomposition and off-grid SBL.
     %where decomposition and off-grid SBL are collectively tested. Our approach offers better accuracy and lower runtime.
\end{itemize}
In short, our algorithm estimates parameters from Kronecker-structured multidimensional signals, tackling grid mismatch and high complexity through two key techniques: decomposition and off-grid SBL. These techniques are of independent interest and can be applied separately, depending on the specific signal model. Compared to our preliminary work \cite{he2023kronecker}, in this work, we make novel contributions: $(i)$ a new off-grid Bayesian algorithm to handle the grid mismatch issue, $(ii)$ theoretical analyses for both decomposition and off-grid SBL, and $(iii)$ comprehensive numerical results that jointly evaluate the decomposition and off-grid for IRS channel estimation. 

\emph{Notation:} We use $[I]$ to denote the set $\{1,2,\cdots,I\}$ and the symbols $\otimes$, $\odot$, $\circ$, and $\times_i$ to denote Kronecker, Khatri-Rao, tensor outer product, and tensor $i$th mode product, respectively.

%Boldface small letters denote vectors, and boldface capital letters denote matrices. The symbol $\|\cdot\|_p$, $\|\cdot\|_\mathrm{F}$, and $\|\cdot\|$ denote the vector $\ell_p$ norm, the matrix Frobenius norm, and the matrix spectral norm, respectively. The symbols $(\cdot)^\mathsf{T}$, $(\cdot)^*$, $(\cdot)^\mathsf{H}$, and $\trace(\cdot)$ are the matrix operations of transpose, conjugate, conjugate transpose, and trace, respectively. If the argument is a vector, operator $\diag(\cdot)$ returns a diagonal matrix with the argument along the diagonal, and it returns a vector of its diagonal entries if the argument is a matrix. Also, $\otimes$, $\odot$, and $\circ$ represent the Kronecker product, the Khatri-Rao product, and the tensor outer product, respectively. We use $\mathcal{N}(\bm a,\bm B)$ and $\mathcal{CN}(\bm a,\bm B)$ to denote the real Gaussian distribution and the complex Gaussian distribution with mean $\bm a$ and covariance $\bm B$, respectively. The set of real, complex matrices of size $M\times N$ is represented by $\mathbb{R}^{M\times N}$ and $\mathbb{C}^{M\times N}$, respectively. We use $[I]$ to denote the set $\{1,2,\cdots,I\}$ and $i\in[I]$ means $\forall \; i=1,2,\cdots,I$.

\section{Off-grid Sparse Recovery Algorithm for Kronecker-structured Measurements}\label{sec.kronecker_re}

In this section, we study the parameter estimation problem with Kronecker-structured measurements. The signal model is
\begin{equation}\label{eq.Kro_generate_data}
    \bar{\bm y} = \otimes_{i=1}^I\bar{\bm H}_{i,\bar{\bm \psi}_i} \bar{\bm x}_i+\bar{\bm n} = \otimes_{i=1}^I \bm y_i +\bar{\bm n},
\end{equation}
where the noise term in $\bar{\bm n}$ need not be Kronecker-structured. For $i\in [I]$, the matrix $\bar{\bm H}_{i,\bar{\bm\psi}_i}$ is parameterized by $\bar{\bm \psi}_i\coloneq\left[\bar{\psi}_{i,1},\ldots,\bar{\psi}_{i,S_i}\right]^\mathsf{T} \in \mathbb{R}^{S_i}$ as follows,
\begin{equation*}
    \bar{\bm H}_{i,\bar{\bm\psi}_i}\coloneq\begin{bmatrix}
        \bm h_i(\bar{\psi}_{i,1})&\cdots&\bm h_i(\bar{\psi}_{i,S_i})
    \end{bmatrix}\in \mathbb{C}^{M_i \times S_i},
\end{equation*}
where $\bm h_i\in\mathbb{C}^{M_i}$ is a \emph{known} and \emph{continuous} column function. The scalar $S_i$ is the number of unknowns in $\bar{\bm H}_{i,\bar{\bm\psi}_i}$. We assume 
%$-\infty< \psi_{i,\mathrm{l}} \leq \bar{\psi}_{i,1} < \cdots < \bar{\psi}_{i,S_i} \leq \psi_{i,\mathrm{r}} < +\infty$, where 
$\bar{\psi}_{i,s}\in[\psi_{i,\mathrm{l}},\psi_{i,\mathrm{r}}]$, a \emph{known} compact range of the unknown parameters, and the goal is to estimate $\bar{\bm\psi}_i$ and $\bar{\bm x}_i$ from \eqref{eq.Kro_generate_data}.

To ensure identifiability of $\bar{\bm \psi}_i$, we assume $\bm h_i(\psi_p)\neq \bm h_i(\psi_q)$ for any $\psi_p \neq \psi_q$. Identificability of $\bar{\bm x}_i$ is limited by the Kronecker structure, i.e., for scalars $\{\alpha_i\}_{i=1}^I$ with $\prod_{i=1}^I\alpha_i=1$, the set of vectors $\{\bar{\bm x}_i\}_{i=1}^I$ and $\{\alpha_i\bar{\bm x}_i\}_{i=1}^I$ both result in $\bar{\bm y}$ when combined with a given noise vector $\bar{\bm n}$. However, in many applications (e.g., channel estimation \cite{he2022structure,zhou2017low}), the goal is to recover the solution up to a scaling factor, as we later elaborate in Sec.\ref{sec.channelmodel}. Therefore, we aim to jointly obtain $\bar{\bm \psi}_i$ and the coefficient $\bar{\bm x}_i$ up to scaling ambiguities, given $i)$ measurement $\bar{\bm y}$, $ii)$ vector function $\bm h_i$, and $iii)$ range $[\psi_{i,\mathrm{l}},\psi_{i,\mathrm{r}}]$ for $i\in[I]$. 

We devise a two-step solution: the first step decomposes~\eqref{eq.Kro_generate_data} into $I$ subproblems, each estimating $\bar{\bm \psi}_i$ and $\bar{\bm x}_i$, and the second step solves these subproblems using a gridless approach. 

\subsection{Step 1: Decomposition-based Algorithm}\label{sec.decomposition}

To develop the decomposition algorithm, 
%we start with the noiseless set of linear equations by setting $\bar{\bm n}=\bm 0$,
% \begin{equation}\label{eq.problem_basic_noiseless}
%     \bar{\bm y} = \otimes_{i=1}^I\bar{\bm H}_{i,\bar{\bm \psi}_i} \bar{\bm x}_i.
% \end{equation}
% We 
we use Lemma~\ref{lmm.kron_equations_separable} for the noiseless set of linear equations, $\bar{\bm y} = \otimes_{i=1}^I\bar{\bm H}_{i,\bar{\bm \psi}_i} \bar{\bm x}_i$. 
\begin{lemma}\label{lmm.kron_equations_separable}
    \cite[Lemma 4]{he2023bayesian} Consider linear equations $\bm y_1 \otimes \bm y_2=\left(\bm H_1 \otimes \bm H_2 \right)\left(\bm x_1\otimes \bm x_2 \right) \neq \bm 0$. Solving for $\bm x_1\otimes\bm x_2$ from the equations is equivalent to solving for  $\bm x_1$ and $\bm x_2$ from $\bm H_1 \left(\alpha \bm x_1\right) = \bm y_1$ and $\bm H_2 \left(\alpha^{-1}\bm x_2\right) = \bm y_2 $, for any scalar $\alpha\neq 0$.%, accounting for the scaling ambiguity.
\end{lemma}
Lemma~\eqref{lmm.kron_equations_separable} indicates that we can estimate individual vectors $\bm x_1$ and $\bm x_2$, up to a scaling ambiguity $\alpha$. Therefore, if $\bar{\bm y}$ is split into $I$ low-dimensional vectors $\{\hat{\bm y}_i \in \mathbb{C}^{M_i}\}_{i=1}^I$, then~\eqref{eq.Kro_generate_data} in the noiseless case ($\bar{\bm n}=\bm 0$) can be decomposed into $I$ subproblems, each with ambiguity $\{\alpha_i \neq 0\}_{i=1}^I$ with $\prod_{i=1}^I \alpha_i = 1$. This approach allows solving for $\bar{\bm x}_i$ individually, rather than jointly. We now discuss the decomposition of $\bar{\bm y}$ into low-dimensional vectors ${\hat{\bm y}_i}$'s, in both noiseless and noisy cases.

\subsubsection{Noiseless Setting and Higher-order Singular Value Decomposition (HOSVD)}
In the noiseless case, we aim to find $\{\hat{\bm y}_i\}_{i=1}^I$ such that $\bar{\bm y} = \bm y=\otimes_{i=1}^I \hat{\bm y}_i$. This can be achieved using HOSVD applied to the tensor representation of $\bar{\bm y}$. Using \eqref{eq.kro_basic}, $\bar{\bm y}$ can be represented as an $I$th order tensor\footnote{From \cite{cichocki2015tensor}, $\vect(\mathcal{\bm Y}) = \vect(\circ_{i=1}^I \bm y_i) = \otimes_{i=I}^1 \bm y_i$ where the subscript is descending. For simplicity, we use ascending subscripts in the tensor outer product, resulting in $\mathcal{\bm Y}$ and $\bar{\bm y}$ containing identical entries, albeit reordered.} $\mathcal{\bm Y} = \circ_{i=1}^I \bm y_i \in \mathbb{C}^{M_1 \times \cdots \times M_I}$, where $\circ$ is the tensor outer product. Then, its $i^*$th mode matricization is 
\begin{equation}\label{eq.matricization}
    \bm Y_{(i^*)} = \bm y_{i^*}\left( \left(\otimes_{i=I}^{i^*+1}\bm y_i\right) \otimes \left(\otimes_{i=i^*-1}^1 \bm y_i  \right)\right)^\mathsf{T},
\end{equation}
where $(\cdot)^\mathsf{T}$ is the transpose operator. Now, an estimate $\hat{\bm y}_{i^*}$ of $\bm y_{i^*}$ up to scaling ambiguities is the left leading singular vector $\bm e_{i^*}$ of the rank-one matrix $\bm Y_{(i^*)}$, i.e., $\hat{\bm y}_{i^*}=\bm e_{i^*}$, for $i^*\in [I-1]$. Also, the estimate $\hat{\bm y}_{I}$ is $\bm e_I$ multiplied by the leading singular value of $\bm Y_{(I)}$, ensuring $\otimes_{i=1}^I\hat{\bm y}_i=\otimes_{i=1}^I\bm y_i$. The decomposition is called the HOSVD, assuming a multilinear rank of $(1,\cdots,1)$ due to the Kronecker structure~\cite{de2000multilinear,cichocki2015tensor,zhang2018tensor}.

\subsubsection{Noisy Case and Truncated HOSVD}\label{sec.trun_hosvd}
Extending to the noisy setting, the decomposition step becomes 
\begin{equation}\label{eq.noisy_decom}
    \{\hat{\bm y}_i\}_{i=1}^I = \underset{\{\bm z_i\in\mathbb{C}^{M_i}\}_{i=1}^I}{\arg\min}\| \bar{\bm y} - \otimes_{i=1}^I \bm z_i \|_2,
\end{equation}
where $\|\cdot\|_2$ is the vector $\ell_2$ norm. We see that~\eqref{eq.noisy_decom} is the same as seeking a tensor $\hat{\mathcal{\bm Y}}=\circ_{i=1}^I \bm z_i$ with multilinear rank $(1,\cdots,1)$ from measurement tensor $\bar{\mathcal{Y}}$ obtained from $\bar{\bm y}$~as
\begin{equation}\label{prob.hosvd}
    \underset{\hat{\mathcal{\bm Y}}}{\min}\| \bar{\mathcal{\bm Y}} - \hat{\mathcal{\bm Y}} \|_\mathrm{F}\ \ \text{s.t.}\ \;\text{multilinear rank of}\ \hat{\mathcal{\bm Y}}\ \text{is}\ (1,\cdots,1),
\end{equation}
where $\|\cdot\|_\mathrm{F}$ is the Frobenius norm. Unlike the noiseless case, here the $i^*$th mode matricization $\bar{\bm Y}_{(i^*)}$ of $\bar{\mathcal{\bm Y}}$ is not rank-one due to noise. We solve \eqref{prob.hosvd} through the truncated HOSVD, where only the left leading singular vector is selected. We obtain  $\hat{\mathcal{\bm Y}} = \xi \times_1 \bm e_1\cdots\times_I\bm e_I$ and $\xi = \bar{\mathcal{\bm Y}} \times_1 \bm e_1^\mathsf{H}\cdots \times_I \bm e_I^\mathsf{H}$, where $\bm e_{i^*}$ is the left leading singular vector of the $i^*$th mode matricization $\bar{\bm Y}_{(i^*)}$ of $\bar{\mathcal{Y}}$ \cite{balda2016first}. %  and $\xi$ is a scalar given by $ \xi = \bar{\mathcal{\bm Y}} \times_1 \bm e_1^\mathsf{H}\cdots \times_I \bm e_I^\mathsf{H}$.
% \begin{equation}
%     \xi = \bar{\mathcal{\bm Y}} \times_1 \bm e_1^\mathsf{H}\cdots \times_I \bm e_I^\mathsf{H}.
% \end{equation}
Here, operator $\times_i$ is the $i$th tensor mode product and $(\cdot)^\mathsf{H}$ is the conjugate transpose. %For example, $\bar{\mathcal{\bm Y}} \times_1 \bm e_1^\mathsf{H}=\bm e_1^\mathsf{H}\bar{\bm Y}_{(1)}$. 
Then, a solution to \eqref{eq.noisy_decom} is $\hat{\bm y}_{i}=\bm e_i$ for $i\in [I-1]$ and $\hat{\bm y}_{I}=\xi\bm e_I$.% ensures $\{\hat{\bm y}_i\}_{i=1}^I$ is a solution to \eqref{prob.hosvd}.%, naturally resolving the nonuniqueness in the decomposing $\bar{\bm y}$.

\subsubsection{A Low-complexity Approximation}
When $I$ and $M_i$ are large, HOSVD can become computationally intensive due to the singular value decomposition (SVD) needed to obtain $\bm e_i$ for $i\in [I-1]$. Hence, we offer a low-complexity method using recursive SVD-based rank-one approximations, 
\begin{equation}\label{prob.bidecom}
(\hat{\bm y}_{i},\bar{\bm y}_{i})=\underset{\substack{(\bm{z}_i\in\mathbb{C}^{M_i},\bar{\bm z}_i),\|\bm z_i\|_2=1}}{\arg\min}\|\bar{\bm y}_{i-1} - \bm z_i\otimes\bar{\bm z}_i\|_2,
\end{equation}
for $i\in[I-1]$ where $\bar{\bm y}_0=\bar{\bm y}$ and $\bar{\bm y}_{I-1}=\hat{\bm y}_I$. For example, we consider the case when $i=1$. We rearrange $\bar{\bm y}$ as $\bar{\bm Y}\in\mathbb{C}^{\bar{M}/M_1\times M_1}$ where $\vect(\bar{\bm Y})=\bar{\bm y}$. Since $\bm z_i\otimes \bar{\bm z}_i=\vect(\bar{\bm z}_i\bm z_i^\mathsf{T})$, \eqref{prob.bidecom} is equivalent to a rank-one approximation that minimizes $\|\bar{\bm Y} - \bar{\bm z}_i\bm z_i^\mathsf{T}\|_\mathrm{F}$, and $\hat{\bm y}_i$ is the leading singular vector of~$\bar{\bm Y}$. 

Compared to HOSVD, here, the problem dimension decreases with $i$ as $\bar{\bm z}_i\in\mathbb{C}^{\prod_{j>i}M_j}$, and the overall complexity is dominated by the first step, i.e., $i=1$. Besides, in the noiseless case,~\eqref{prob.bidecom} and HOSVD yield the same solution. 

Combining the decomposition step for $\bm y$ obtaining $\{\hat{\bm y}_i\}_{i=1}^I$ with Lemma~\ref{lmm.kron_equations_separable}, we  break down the original $\bar{M}$-dimensional problem into $I$ subproblems of dimensions $\{M_i\}_{i=1}^{I}$,
\begin{equation}\label{eq.deco_linear_inversion}
    \hat{\bm y}_i = \bar{\bm H}_{i,\bar{\bm \psi}_i} \bar{\bm x}_i + \bar{\bm n}_i,\ \ i\in[I],
\end{equation}
which can be solved in parallel. Here, we assume $\alpha_i = 1$ without loss of generality, as we seek solutions up to a scaling factor. The decomposition fully exploits the Kronecker structure in the measurements, aiding denoising (see Sec.~\ref{sec:estimation_acc}) and reducing the complexity (see Sec.~\ref{sec.extension}). Before presenting these analyses, we first develop an algorithm to estimate $\bar{\bm \psi}_i$ and $\bar{\bm x}_i$ from \eqref{eq.deco_linear_inversion} for a given $i$.

\begin{algorithm}[t]
\caption{OffSBL}
\label{al.offSBL}
\begin{algorithmic}[1]
\Statex \textit {\bf Input:} Measurement $\hat{\bm y}$, the number of initial grids $N$, the range $[\psi_{\mathrm{l}},\psi_{\mathrm{r}}]$, and thresholds $\epsilon_1,\epsilon_2<1$

\State Set $r=0$, initialize $\bm \gamma^{(0)}=\bm 1$, and initialize $\bm \psi^{(0)}$ with uniform samples $\{\theta_n\}_{n=1}^N$ from $[\psi_{\mathrm{l}},\psi_{\mathrm{r}}]$. 
\Repeat

\State Compute $\bm \mu_{\mathsf{x}}$ and $\bm \Sigma_{\mathsf{x}}$ using~\eqref{eq.post_meva}  
% \Comment{E-step}

\State Update $\bm \gamma^{(r+1)}$ using~\eqref{eq.qfunc}
\State Set $t=0$, initialize $\bm \psi^{(r,t)} = \bm \psi^{(r)}$

\Repeat

\State Obtain ${\psi}_{n^*}^{(r,t+1)}$ by solving~\eqref{eq.single_op}, for $n^*\in[N]$

\State Set $t=t+1$
\Until {$\|\bm \psi^{(t+1)} - \bm \psi^{(t)}\| < \epsilon_1$}

\State Let $\bm \psi^{(r+1)} = \bm \psi^{(r,t)}$
\State Obtain $(\sigma^2)^{(r+1)}$ using~\eqref{eq.noise_est} and set $r=r+1$
\Until{$\|\bm \gamma^{(r+1)} - \bm \gamma^{(r)}\|_2/\|\bm \gamma^{(r)}\|_2 < \epsilon_2$}
\Statex \textit {\bf Output:} Estimated $\bm x = \bm \mu_{\mathsf{x}}$.
\end{algorithmic}
\end{algorithm}

\begin{algorithm}[t]
\caption{Decomposition-based SBL (dSBL)}
\label{al.dKroSBL}
\begin{algorithmic}[1]
\Statex \textit {\bf Input:} Measurement $\bar{\bm y}$, dictionaries $\bm H_i$ for $i\in[I]$
% \Statex \textit {\bf Output:} Sparse vector $\bm x $

\For {$i\in[I]$}
\State Solve~\eqref{prob.hosvd} or~\eqref{prob.bidecom} to obtain $\hat{\bm y}_i$
\State Solve~\eqref{eq.deco_linear_inversion} for $\bm x_i$ using Algorithm~\ref{al.offSBL}
\EndFor
\Statex \textit {\bf Output:} For $i\in[I]$, $\bar{\bm x}_i$ corresponds to the nonzero entries of $\bm x_i$ and $\bar{\bm \psi}_i$ corresponds to the support of $\bm x_i$.
\end{algorithmic}
\end{algorithm}

\subsection{Step 2: Off-grid SBL-based Estimation Algorithm}\label{sec.mbem}
In each dimension, the subproblem takes the general form of $\hat{\bm y} = \bar{\bm H}_{\bar{\bm \psi}} \bar{\bm x} + \bar{\bm n}$ for a nonlinear function $\bm h$ when entries of $\bar{\bm \psi}$ belong to $[\psi_{\mathrm{l}},\psi_{\mathrm{r}}]$, where we drop the dimension index $i$. Here, we adopt BEM by discretizing $[\psi_{\mathrm{l}},\psi_{\mathrm{r}}]$ with grid points $\bm \psi\in\mathbb{R}^N$ with the $n$th grid point $\psi_n$, forming the BEM dictionary 
\begin{equation*}%\label{eq.bem_steering}
    \bm H(\bm \psi)\coloneq\begin{bmatrix}
        \bm h(\psi_{1})&\cdots&\bm h(\psi_{N})
    \end{bmatrix}\in \mathbb{R}^{M \times N}.
    \end{equation*}
%where $\psi_n$ is the $n$th entry of $\bm \psi$. 
This leads to the BEM with coefficient vector ${\bm x}$ as
    \begin{equation}
    \hat{\bm y} = {\bm H}(\bm \psi){\bm x} + {\bm n},\label{eq.problem_basic_bem}
\end{equation}
%This leads to the BEM with coefficient vector ${\bm x}$ as
% \begin{equation}\label{eq.problem_basic_bem}
%     \hat{\bm y} = {\bm H}(\bm \psi){\bm x} + {\bm n}.
% \end{equation}
Only a few entries of $\bm \psi$ correspond to the true parameters $\bar{\bm \psi}$, making $\bm x$ sparse. However, solving \eqref{eq.problem_basic_bem} with a fixed $\bm \psi$ via standard sparse recovery algorithms leads to grid mismatch. Thus, we jointly estimate $\bm \psi$ and sparse $\bm x$ using the SBL framework. Then, the nonzero entries of $\bm x$ and their corresponding $\psi_n$'s are estimates of $\bar{\bm x}$ and $\bar{\bm \psi}$, respectively. 

SBL adopts a fictitious zero mean complex Gaussian distribution $\mathcal{CN}(\bm x|\bm 0,\bm \Gamma)$ as prior on the sparse vector ${\bm x}$ with an unknown diagonal covariance matrix $\bm\Gamma$. Let $\bm\Gamma= \diag(\bm \gamma)\in\mathbb{R}^{N\times N}$ with the diagonal entries $\bm \gamma\in\mathbb{R}^{N}$. We assume Gaussian noise $\bm n \sim \mathcal{CN}(0, \sigma^2\bm I)$ with unknown variance $\sigma^2$. Using type II ML, we first estimate the hyperparameters $\bm \gamma$, $\bm \psi$, and  $\sigma^2$, and then, the estimate of $\bm x$ is $\arg \max_{\bm x} p(\bm x|\hat{\bm y};\bm \gamma,\bm \psi,\sigma^2)$. The ML estimates of the hyperparameters are 
\begin{equation}\label{eq.ml_problem_para}
    \underset{{\bm \gamma \geq \bm 0,\ \bm \psi\in [\psi_{\mathrm{l}},\psi_{\mathrm{r}}]^N},\ \sigma^2 > 0}{\min} \mathcal{L}\left(\bm \gamma,\bm \psi,\sigma^2\right),
\end{equation}
where $\bm \gamma \geq \bm 0$ indicates that the entries of $\bm \gamma$ are nonnegative. Using the SBL priors, negative log-likelihood function $\mathcal{L}$ is
\begin{equation*}
    \mathcal{L}\left(\bm \gamma,\bm \psi,\sigma^2\right)=\log p(\hat{\bm y};\bm \gamma,\bm \psi,\sigma^2)=\log |\bm \Sigma_{\mathsf{y}}|+\trace \left(\hat{\bm y}^\mathsf{H} \bm \Sigma_{\mathsf{y}}^{-1} \hat{\bm y}\right),
\end{equation*}
where $\bm \Sigma_{\mathsf{y}} = \sigma^2 \bm I_{{M}} + {\bm H}(\bm \psi) \bm \Gamma {\bm H}(\bm \psi)^\mathsf{H}$ and $\trace(\cdot)$ is the trace operator. We resort to the EM algorithm to solve \eqref{eq.ml_problem_para}. Specifically, the $r$th iteration of~EM~is
\begin{align*}
&\text{\bf E-step:}~Q(\bm \gamma,\bm \psi,\sigma^2)\notag\\
&\quad\quad\quad\quad\quad\quad\quad\coloneq\mathbb{E}_{\bm x|\hat{\bm y};\{\bm \gamma,\bm \psi,\sigma^2\}^{(r)}}\{\log p(\hat{\bm y},\bm x;\bm \gamma,\bm \psi,\sigma^2)\}\notag,\\
&\text{\bf M-step:}\{\bm \gamma,\bm \psi,\sigma^2\}^{(r+1)}=\underset{\substack{\bm \gamma > \bm 0, \bm \psi\in [\psi_{\mathrm{l}},\psi_{\mathrm{r}}]^N}, \sigma^2>0}{\arg\max}\ Q(\bm \gamma,\bm \psi,\sigma^2).  \label{eq.mstep_basic}  
\end{align*}
Here, $\bm \gamma >\bm0$ means the entries of $\bm \gamma$ should be positive to avoid degenerate distributions. Further, we note that 
\begin{multline*}
    Q(\bm \gamma,\bm \psi,\sigma^2)=\mathbb{E}_{\bm x|\hat{\bm y};\{\bm \gamma,\bm \psi,\sigma^2\}^{(r)}}\{\log p(\hat{\bm y}|\bm x;\bm \psi,\sigma^2)\}\\+\mathbb{E}_{\bm x|\hat{\bm y};\{\bm \gamma,\bm \psi,\sigma^2\}^{(r)}}\{\log p(\bm x;\bm \gamma)\},
\end{multline*}
and thus, the optimization problem in the M-step is separable in $\bm \gamma$ and $\{\bm \psi,\sigma^2\}$. The optimization problem in $\bm\gamma$~is
\begin{equation}\label{eq.qfunc}
\bm \gamma^{(r+1)}
= \underset{\substack{\bm \gamma >\bm 0}}{\arg\min}
\log|\diag(\bm\gamma)|+(\bm d^{(r)})^{\mathsf{T}}\bm\gamma^{-1} = \bm d^{(r)},
\end{equation}    
with $\bm d^{(r)} = \diag(\bm \Sigma_{\mathsf{x}}+\bm \mu_{\mathsf{x}}\bm \mu_{\mathsf{x}}^\mathsf{H})$ where $\diag(\cdot)$ returns the diagonal entries of the matrix input, and $\bm \gamma^{-1}$ representing element-wise inversion. Here, $\bm \mu_{\mathsf{x}}$ and $\bm \Sigma_{\mathsf{x}}$ are the mean and variance of conditional distribution $p(\bm x|\hat{\bm y};\{\bm \gamma,\bm \psi,\sigma^2\}^{(r)})$, respectively,
\begin{equation}
\label{eq.post_meva}
\begin{aligned}
    \bm \mu_{\mathsf{x}} &= \left(\sigma^{-2}\right)^{(r)}\bm \Sigma_{\mathsf{x}}{\bm H}(\bm \psi^{(r)})^\mathsf{H}\hat{\bm y},\\
    \bm \Sigma_{\mathsf{x}} &=  \left[\left(\sigma^{-2}\right)^{(r)}{\bm H}(\bm \psi^{(r)})^\mathsf{H}{\bm H}(\bm \psi^{(r)}) + \diag(\bm \gamma^{(r)})^{-1}\right]^{-1}. 
\end{aligned}
\end{equation}
Optimizing $\bm \psi$ and $\sigma^2$ in the M-step yields
\begin{equation*}%\label{eq.m_step_noise_grid}
    \{\bm \psi,\sigma^2\}^{(r+1)} = \underset{\bm \psi\in [\psi_{\mathrm{l}},\psi_{\mathrm{r}}]^N,\ \sigma^2 > 0}{\arg\min} N\log \sigma + \frac{1}{2\sigma^2} g(\bm \psi),
\end{equation*}
where $g(\bm \psi)\coloneq\|\hat{\bm y} - \bm H(\bm \psi)\bm \mu_{\mathsf{x}} \|_2^2+\trace(\bm \Sigma_{\mathsf{x}}\bm H(\bm \psi)^\mathsf{H}\bm H(\bm \psi))$, is independent of $\sigma^2$. Given $\bm \psi^{(r+1)}$, we update $\sigma^2$ as
\begin{equation}\label{eq.noise_est}
    (\sigma^2)^{(r+1)} = g(\bm \psi^{(r+1)})/N.
\end{equation}
Further, $\bm \psi^{(r+1)}= {\arg\min}_{\bm \psi}\ g(\bm \psi)$ simplifies~to
\begin{multline}\label{prob.op_wrt_para}
    \bm \psi^{(r+1)} =\underset{\bm \psi\in [\psi_{\mathrm{l}},\psi_{\mathrm{r}}]^N}{\arg\min}\trace\left(\bm H(\bm \psi) \bm \Sigma \bm H(\bm \psi)^\mathsf{H}\right)\\
    -2\mathsf{Re}\left\{\trace\left(\bm M\bm H(\bm \psi)\right)\right\},
\end{multline}
where $\bm \Sigma \coloneq \bm \Sigma_{\mathsf{x}}+\bm \mu_{\mathsf{x}}\bm \mu_{\mathsf{x}}^\mathsf{H}$ and $\bm M = \bm \mu_{\mathsf{x}}\hat{\bm y}^\mathsf{H}$.
We use the alternating minimization method to solve \eqref{prob.op_wrt_para}, where we alternatively optimize one entry of $\bm \psi$ while keeping all others fixed. The $t$th iteration of the alternating minimization method updates the $n^*$th variable ${\psi}_{n^*}$ by minimizing
\begin{equation*}
    f_{n^*}(\psi_{n^*}) = 2\mathsf{Re}\left\{\bm v_{n^*}^\mathsf{H} \bm h(\psi_{n^*})\right\} + \bm \Sigma_{n^*,n^*}\bm h(\psi_{n^*})^\mathsf{H}\bm h(\psi_{n^*}).
\end{equation*}
Here, $\bm v_{n^*}\in \mathbb{C}^{M}$ is defined as
\begin{equation*}
    \bm v_{n^*}= \!\sum_{n=1}^{n^*-1}\!\bm \Sigma_{n^*,n}\bm h (\psi_n^{(r,t+1)}) +\!\! \sum_{n=n^*+1}^N\!\!\bm \Sigma_{n^*,n}\bm h (\psi_n^{(r,t)}) - \bm M_{n^*,:}^\mathsf{H},
\end{equation*}
where $\bm M_{n^*,:}$ is the $n^*$th row of $\bm M$, and $\bm \Sigma_{n^*,n^*}$ is the $(n^*,n^*)$th entry of $\bm \Sigma$. Hence, the $t$th alternative minimization iterate $\psi_{n^*}^{(r,t+1)}$ in the $r$th EM iteration for index $n^* \in [N]$ is
\begin{equation}\label{eq.single_op}
        \psi_{n^*}^{(r,t+1)} = \underset{ \frac{\psi_{n^*}^{(r,t)}+\psi_{n^*-1}^{(r,t)}}{2}\leq \psi\leq\frac{\psi_{n^*}^{(r,t)}+\psi_{n^*+1}^{(r,t)}}{2}}{\arg\min}\ f_{n^*}(\psi).
\end{equation}
We note that assumption $\bm h_i(\psi_p)\neq \bm h_i(\psi_q)$ for any $\psi_p \neq \psi_q$ ensures the solution identifiablity of problem~\eqref{eq.single_op}. We solve~\eqref{eq.single_op} using (one-dimensional) line search. Our off-grid SBL (OffSBL) and the overall \underline{d}ecomposition-based SBL (dSBL) are outlined in Algorithms~\ref{al.offSBL} and~\ref{al.dKroSBL}, respectively.

\section{Theoretical Analysis and Extensions}\label{sec.theoretical}

This section analyzes our dSBL algorithm, covering the decomposition error bound and denoising effect of HOSVD and convergence results of OffSBL. We then present the complexity analysis of dSBL, demonstrating its computational advantage compared to other methods. Finally, we discuss extensions of our algorithms to other similar signal models.

\subsection{Analysis of Decomposition-Based Algorithm}\label{sec:estimation_acc}

% In this section, we discuss the decomposition error bound on $\{\hat{\bm y}_i\}_{i=1}^I$, the denoising in HOSVD, and its implication in low complexity approximation. For simplicity, we assume that $M_i \sim M$, i.e., there exists a measurement level $M$ and constant $C,c>0$ such that $cM \leq M_i \leq CM$ for $i \in [I]$, though generalization is straightforward. Also, the noise term $\bar{\bm n}$ in~\eqref{eq.Kro_generate_data} is assumed to be zero-mean white Gaussian noise with variance $\sigma^2_\mathrm{t}$.

We start with the decomposition error bound, where we quantify the error between the decomposed vectors $\{\hat{\bm y}_i\}_{i=1}^I$ and the true signal components $\{\bm y_i=\bar{\bm H}_{i,\bar{\bm \psi}_i} \bar{\bm x}_i\}_{i=1}^I$. We measure the error as the angle between $\bm y_i$ and $\hat{\bm y}_i$, accounting for scaling ambiguity in the decomposition step. %The following result offers a probabilistic upper bound on the error.

\begin{thm}[Decomposition Accuracy]\label{thm.angle_denoise}
    Let $\bar{\bm y} = \otimes_{i=1}^I \bm y_i +\bar{\bm n}\in\mathbb{R}^{\bar{M}}$ be the noisy measurement as given in \eqref{eq.Kro_generate_data}, where $\bm y_i\in\mathbb{R}^{M_i}$ and $\bar{\bm n}$ has independent zero-mean Gaussian entries with variance $\sigma^2_\mathrm{t}$.   
    Suppose the signal satisfies
    \begin{equation*}
        \lambda^2=\|\otimes_{i=1}^I \bm y_i\|^2_2/\sigma_\mathsf{t}^2 \geq C_{\mathrm{gap}} ( \sqrt{\bar{M}}+\max_{1\leq i\leq I} M_i),
    \end{equation*}
    for a large constant $C_\mathrm{gap}>0$. Then, there exist constants $c,C>0$ such that with probability at least $1-C \exp \{-cM_i\}$,
    \begin{equation*}
        \sin (\vartheta_i) \leq  C \left( \sqrt{M_i}\lambda^{-1} +\sqrt{\bar{M}}\lambda^{-2} \right),\quad i\in[I],
    \end{equation*}
     where $\vartheta_i\coloneq\arccos\left(|\bm y_i^\mathsf{T}\hat{\bm y_i}|/(\|\bm y_i\|_2\|\hat{\bm y_i}\|_2)\right)$ is the angle between $\bm y_i$ and its estimate $\hat{\bm y}_i$ obtained by solving~\eqref{eq.noisy_decom} and~\eqref{prob.hosvd}.
\end{thm}

\begin{proof}
See Appendix~\ref{appe.angle_and_denoising}. 
\end{proof}
We note that $\lambda/\bar{M}$ represents the signal-to-noise ratio (SNR) of our measurement model, and a higher SNR (i.e., a larger $\lambda$) improves estimation accuracy, as expected. As $\lambda$ goes to $\infty$ (noiseless case), the error bounds approach zero. Conversely, when the signal strength is insufficient to meet the required condition, there is no consistent estimator for $\bm y_i$'s~\cite{zhang2018tensor}. 

While the decomposition accuracy reflects how well $\hat{\bm y_i}$ aligns with the true signal $\bm y_i$, we can also access the noise level after decomposition. We next quantify the denoising effect of the decomposition step, which refers to noise reduction in the measurements, i.e., $\Vert\otimes_{i=1}^I \hat{\bm y}_i-\otimes_{i=1}^I \bm y_i\Vert_2^2$ is expected to be smaller than $\mathbb{E}\{\Vert\bar{\bm{n}}\Vert_2^2\}=\sigma_t^2\bar{M}$, as summarized in the following result. %We denote $\bar{M}/M_i$ as $M_{-i}$.

\begin{thm}[Denoising Effect]\label{thm.denoising}
    Let $\bar{\bm y} = \otimes_{i=1}^I \bm y_i +\bar{\bm n}\in\mathbb{R}^{\bar{M}}$ denote the noisy measurement as in \eqref{eq.Kro_generate_data}, where $\bm y_i\in\mathbb{R}^{M_i}$ and $\bar{\bm n}$ has independent zero-mean Gaussian entries with variance $\sigma^2_\mathrm{t}$. Let $\hat{\bm y}_i$ denote the estimate of $\bm y_i$ obtained by solving~\eqref{eq.noisy_decom} and~\eqref{prob.hosvd}. Then, the estimates satisfy%$\Vert\otimes_{i=1}^I \hat{\bm y}_i-\otimes_{i=1}^I \bm y_i\Vert_2$ is probabilistically bounded as
    \begin{multline}\label{eq.probabilistic}
        \|\otimes_{i=1}^I \hat{\bm y}_i-\otimes_{i=1}^I \bm y_i\|_2 \\\leq \sigma_\mathrm{t}\left(2\sum_{i=1}^I \left[3\sqrt{M_i} + \sqrt{\bar{M}/M_i}\right]+1 + 2\sqrt{\max_{1\leq i\leq I}M_i}\right),
    \end{multline}
    with probability exceeding $1-3\sum_{i=1}^Ie^{-M_i}$. Moreover, 
    \begin{equation}\label{eq.appro_denoise_hosvd}
    \mathbb{E}\left\{\Vert\otimes_{i=1}^I \hat{\bm y}_i-\otimes_{i=1}^I \bm y_i\Vert_2^2\right\} \approx \sigma_t^2\left(\sum_{i=1}^I M_i + 1 - I \right).
    \end{equation}
\end{thm}

\begin{proof}
See Appendix \ref{appe.thm_denoising}.    
\end{proof}

To gain insights from Theorem~\ref{thm.denoising}, suppose that $M_i=\mathcal{O}(M)$ for $i\in [I]$ for some value $M$. Then, \eqref{eq.probabilistic} shows that the noise level $\Vert\otimes_{i=1}^I \hat{\bm y}_i-\otimes_{i=1}^I \bm y_i\Vert^2_2$ in the decomposed signal compared to the noisy signal $\bar{\bm y}$ reduces from $\mathcal{O}(M^I\sigma_t^2)$ to $\mathcal{O}(M^{I-1}\sigma_t^2)$ after HOSVD.  Besides, from \eqref{eq.appro_denoise_hosvd}, the average noise level can be approximated as $\mathcal{O}(MI\sigma^2_t)$ when $M>I$. Specifically, the noise level $\mathbb{E}\{\Vert\bar{\bm{n}}\Vert^2\}=\sigma^2_\mathrm{t}\bar{M}$ reduces approximately by 
\begin{equation*}
    \frac{\mathbb{E}\left\{\Vert\otimes_{i=1}^I \hat{\bm y}_i-\otimes_{i=1}^I \bm y_i\Vert_2^2\right\}}{\mathbb{E}\{\Vert\bar{\bm{n}}\Vert_2^2\}} \approx \frac{\sum_{i=1}^I M_i + 1 - I }{\bar{M}}<1.
\end{equation*}
% \begin{proposition}\label{prop.denoise_hosvd}
%     Suppose that the noiseless signal $\otimes_{i=1}^I \bm y_i$ is deteriorated by zero-mean white Gaussian noise $\bar{\bm{n}}$ with variance $\sigma^2_\mathrm{t}$. Let $\otimes_{i=1}^I\hat{\bm y}_i$ be the estimate of $\otimes_{i=1}^I \bm y_i$ using HOSVD as in Sec.~\ref{sec.trun_hosvd}. Then $\mathbb{E}\{\Vert\otimes_{i=1}^I \hat{\bm y}_i-\otimes_{i=1}^I \bm y_i\Vert^2\}$ can be approximated as
%     \begin{equation}
%     \mathbb{E}\left\{\Vert\otimes_{i=1}^I \hat{\bm y}_i-\otimes_{i=1}^I \bm y_i\Vert^2\right\} \approx \sigma^2_\mathrm{t}\left( \sum_{i=1}^I M_i + 1 - I \right).
%     \end{equation}
% \end{proposition}

The probabilistic bound in Theorem \ref{thm.denoising} can also be interpreted as an error bound for HOSVD. Consider the simplest case $I=2$ and fix $\bar{M}=M_1M_2$ and $\sigma_\mathrm{t}$. The upper bound in~\eqref{eq.probabilistic} (with respect to $\sigma_\mathrm{t}$) can be bounded from below as
\begin{multline*}
    8(\sqrt{M_1} + \sqrt{M_2})+1 + 2\sqrt{\max\{M_1,M_2\}}\\\geq
    8\!\left[\min_{\substack{M_1,M_2\\M_1M_2=\bar{M}}}\!\sqrt{M_1} + \sqrt{M_2}\right]+1 + 2\bar{M}^{1/4}
    = 18\bar{M}^{1/4}+1,
\end{multline*}
where equality is achieved when $M_1=M_2$. Therefore, the bound in~\eqref{eq.probabilistic} is minimized when the vectors $\bm y_i$'s have the same size.

We present the next corollary on the low complexity approximation, obtained by setting $I=2$ and $M_2=\bar{M}/M_1$ in Theorem \ref{thm.denoising}, noting that it has the same first step as HOSVD.
\begin{coro}\label{coro.denoise_deco_1}
   Under the assumptions of Theorem~\ref{thm.denoising}, if $\bar{\bm y}_{1}$ and $\hat{\bm y}_1$ are obtained from $\bar{\bm y}$ using the low complexity approximation~\eqref{prob.bidecom}, with probability at least $1-3(e^{-M_1}+e^{-\bar{M}/M_1})$, 
    \begin{equation*}
        \|\hat{\bm y}_1\otimes \bar{\bm y}_1 - \bar{\bm y}\|_2 \leq \sigma_\mathrm{t}\left( 1 + 18\sqrt{\max\left\{M_1,{\bar{M}/M_1}\right\}}   \right),
    \end{equation*}
    and $\mathbb{E}\left\{\|\hat{\bm y}_1\otimes \bar{\bm y}_1 - \bar{\bm y}\|_2^2\right\} \approx \sigma^2_\mathrm{t}( M_1+\bar{M}/M_1-1)
    $.
    % \begin{equation}
    % \mathbb{E}\left\{\|\hat{\bm y}_1\otimes \bar{\bm y}_1 - \bar{\bm y}\|_2^2\right\} \approx \sigma^2_\mathrm{t}( M_1-1+M_{-1}).
    % \end{equation}
\end{coro}
% \begin{proof}
%     The result is a special case of perturbation analysis of low-rank tensor approximations in \cite{balda2016first} by setting the tensor order $R=2$ and $r$-ranks as $p_1 = p_2 = 1$ in \cite[Eq. (19)]{balda2016first}.
% \end{proof}
Corollary~\ref{coro.denoise_deco_1} shows that the first step of low complexity approximation also aids denoising. 
To intuitively see this, we reorganize $\otimes_{i=1}^I \bm y_i$ into a rank-one matrix $\bm Y:=(\otimes_{i=2}^I\bm y_i) \bm y_1 ^\mathsf{T}$
\begin{equation*}\label{eq.step_1}
    \bar{\bm Y} = \bm Y + \bar{\bm N} = \left( \otimes_{i=2}^I\bm y_i \right)\bm y_1^\mathsf{T} + \bar{\bm N},
\end{equation*}
with $\vect{(\bar{\bm Y})} = \bar{\bm y}$, $\vect(\bm Y)=\otimes_{i=1}^I \bm y_i$, and $\vect{(\bar{\bm N})} = \bar{\bm n}$. Here, the noise term is unstructured $\bar{\bm N}$ and typically has full rank. The first step ($i=1$) of~\eqref{prob.bidecom}, yields $\hat{\bm y}_1$ and $\bar{\bm y}_1$, which estimate $\bm Y$ as $\bar{\bm y}_1\hat{\bm y}_1^{\mathsf{T}}$.  Comparing this estimate with $\bar{\bm Y}$, we observe that $\bar{\bm y}_1\hat{\bm y}_1^{\mathsf{T}}$ preserves the rank-one structure of the signal. It discards the components that violate the rank-one constraint, which are often attributed to noise, effectively performing denoising. However, a drawback of the low-complexity approximation is that the error from one step can propagate to the subsequent steps, as later computations depend on estimates from the previous steps. In contrast, the error in HOSVD is independent in each subspace, leading to a better decomposition but with a higher computation cost.

%To summarize, we discuss the decomposition error bound and demonstrate the denoising effect of the decomposition step, which facilitates our algorithm to achieve a better reconstruction error as an added advantage in addition to the reduced complexity. 

\subsection{Analysis of OffSBL Algorithm}

This section discusses the convergence results for OffSBL in Algorithm~\ref{al.offSBL}. We note that OffSBL is a two-level iterative algorithm, where the outer EM iteration is given by \eqref{eq.post_meva} followed by \eqref{eq.qfunc} and \eqref{eq.noise_est}, and the inner loop updates the grid points via \eqref{eq.single_op}. We first provide the guarantees for the inner loop for a given $r$th EM iteration.

\begin{lemma}[Convergence of Grid Update]\label{lm.convergent_inner}
    Consider the alternating grid update~\eqref{eq.single_op} in the $r$th EM iteration. If $\sup_{\psi\in [\psi_{\mathrm{l}},\psi_{\mathrm{r}}]}\|\bm h(\psi)\|_2^2 <\infty$, the sequence $\{g(\bm \psi^{(r,t)})\}_{t=0}^\infty$ is non-increasing and convergent. Its iterate $\{\bm\psi^{(r,t)}\}_{t=0}^\infty$ adopts at least one limit point.
\end{lemma}

\begin{proof}
First, we note that there always exists an optimal solution to~\eqref{eq.single_op} due to the continuity of the function $f_{n^*}$. The extreme value theorem states that $f_{n^*}(\psi)$ must reach a minimum at least once within the closed and bounded constraint set. The solvability of~\eqref{eq.single_op} further indicates 
    \begin{equation}\label{eq.monotonicity}
        g(\bm \psi^{(r,t+1)})\leq g(\bm \psi^{(r,t)}).
    \end{equation}
    Further, the cost function $g(\bm \psi)$ is bounded from below as
    \begin{align*}
        g(\bm \psi) &=\|\hat{\bm y} - \bm H(\bm \psi)\bm \mu_{\mathsf{x}} \|_2^2+\trace(\bm \Sigma_{\mathsf{x}}\bm H(\bm \psi)^\mathsf{H}\bm H(\bm \psi))\\
        &\geq -2\mathsf{Re}(\trace \left\{ \bm M\bm H(\bm \psi)\right\} )\\
        &\geq - \trace \left\{ \bm M^\mathsf{H} \bm M \right\} - \trace\left\{\bm H^\mathsf{H}(\bm \psi)\bm H(\bm \psi) \right\}  \\ & \geq - \trace \left\{ \bm M^\mathsf{H} \bm M \right\} - N \sup_{\psi\in [\psi_{\mathrm{l}},\psi_{\mathrm{r}}]}\|\bm h(\psi)\|_2^2>-\infty.
    \end{align*}
    % \begin{multline*}
    %     g(\bm \psi) 
    %     = \trace \left\{ (\bm M^\mathsf{H}-\bm H)^\mathsf{H}(\bm M^\mathsf{H}-\bm H) + \bm H \bm \Sigma \bm H^\mathsf{H} \right\} 
    %     \\- \trace \left\{ \bm M^\mathsf{H} \bm M \right\} - 2\trace\left\{\bm H^\mathsf{H}\bm H \right\} 
    %     \\\geq - \trace \left\{ \bm M^\mathsf{H} \bm M \right\} - 2\trace\left\{\bm H^\mathsf{H}\bm H \right\}    \geq - \trace \left\{ \bm M^\mathsf{H} \bm M \right\} - 2N\epsilon.
    % \end{multline*}
    Thus, by the monotone convergence theorem, the sequence $\{g(\bm \psi^{(r,t)})\}_{t=0}^\infty$ converges. The monotonicity~\eqref{eq.monotonicity} also ensures that $\{\bm \psi^{(r,t)}\}_{t=0}^\infty$ belongs to the sublevel set of $\bm \psi^{(r,0)}$.  Since $g(\bm \psi)$ is continuous, its sublevel sets are compact, implying that the sequence $\{\bm \psi^{(r,t)}\}_{t=0}^\infty$ adopts at least one limit point.
\end{proof}

We next prove the convergence of our OffSBL algorithm. For the convergence result, we assume the zero-mean Gaussian noise $\bm{n}$ in~\eqref{eq.problem_basic_bem} has a known variance $\sigma^2 > 0$ and present the convergence based on iterates $\{\bm \gamma^{(r)},\bm \psi^{(r)}\}_{r=0}^\infty$, i.e., we do not update $\sigma^2$ via \eqref{eq.noise_est}, but use the true value of $\sigma^2$ in Algorithm~\ref{al.offSBL}.  This assumption simplifies deriving a lower bound for the negative log-likelihood function~\eqref{eq.ml_problem_para}, which is challenging when $\sigma^2 = 0$ or treated as a variable. Moreover, assuming $\sigma^2 > 0$ is standard in SBL analysis \cite{joseph2019convergence,khanna2022support} and the noiseless settings corresponds the limit where $\sigma^2 \to 0$. Thus, the results below are asymptotically applicable to the noiseless case.
 
\begin{thm}[Convergence Property of OffSBL]\label{thm.convergent_cost}
Consider the problem~\eqref{eq.problem_basic_bem} where $\bm n$ is zero-mean Gaussian noise with known variance $\sigma^2 > 0$, solved using OffSBL in Algorithm~\ref{al.offSBL}. If $\sup_{\psi\in [\psi_{\mathrm{l}},\psi_{\mathrm{r}}]}\|\bm h(\psi)\|_2^2 <\infty$, the cost function sequence $\{\mathcal{L}( \bm \gamma^{(r)},\bm \psi^{(r)} )\}_{r=0}^\infty$ converges monotonically to some value $\mathcal{L}^*$, and the sequence $\{\bm \gamma^{(r)},\bm \psi^{(r)}\}_{r=0}^\infty$ converges to a set $\mathcal{S}^*$ with $\mathcal{L}(\bm \gamma,\bm \psi)=\mathcal{L}^*$ for any $\{\bm \gamma,\bm \psi\} \in \mathcal{S}^*$.
\end{thm}

\begin{proof}
    We have $Q(\bm \gamma^{(r)},\bm \psi^{(r)}) \leq Q(\bm \gamma^{(r+1)},\bm \psi^{(r+1)})$ in the $r$th EM iteration due to Lemma \ref{lm.convergent_inner}. Then, OffSBL is a generalized EM algorithm and the sequence $\{\mathcal{L}( \bm \gamma^{(r)},\bm \psi^{(r)} )\}_{r=0}^\infty$ is nonincreasing \cite{wu1983convergence}. Also, $\mathcal{L}(\bm \gamma,\bm \psi)$ is bounded from below~as
    \begin{equation*}
        \mathcal{L}\left(\bm \gamma,\bm \psi\right)=\log |\bm \Sigma_{\mathsf{y}}|+\trace \left(\hat{\bm y}^\mathsf{H} \bm \Sigma_{\mathsf{y}}^{-1} \hat{\bm y}\right)\geq \log |\bm \Sigma_{\mathsf{y}}|\geq M \log \sigma^2,
    \end{equation*}
    where $\bm \Sigma_{\mathsf{y}}^{-1} = (\sigma^2 \bm I_{{M}} + {\bm H}(\bm \psi) \bm \Gamma {\bm H}(\bm \psi)^\mathsf{H})^{-1}$ is positive definite for $\bm \gamma \geq \bm 0$ and $ \bm \psi\in [\psi_{\mathrm{l}},\psi_{\mathrm{r}}]^N$. The last inequality is because the eigenvalues of $\bm \Sigma_{\mathsf{y}}$ are lower bounded by $\sigma^2$~\cite{he2023bayesian}. Thus, by the monotone convergence theorem,  the sequence $\{\mathcal{L}( \bm \gamma^{(r)},\bm \psi^{(r)} )\}_{r=0}^\infty$ converges to some value $\mathcal{L}^*$. 
    
    Further, the function $\mathcal{L}$ is a coercive function of $\bm \gamma$ and continuous in both $\bm \gamma$ and $\bm \psi$ \cite[Lemma 3]{joseph2019convergence}. Consequently, its sublevel sets are compact. The nonincreasing $\{\mathcal{L}( \bm \gamma^{(r)},\bm \psi^{(r)} )\}_{r=0}^\infty$ indicates that $\{\bm \gamma^{(r)},\bm \psi^{(r)}\}_{r=0}^\infty$ adopts at least one limit point in the level set of $\mathcal{L}^*$, i.e., $\mathcal{S}^*$, which completes the proof.
\end{proof}

We conclude by noting that although the properties of limit points of the EM iterations are unknown, OffSBL offers a sequence over which the negative log-likelihood $\mathcal{L}$ is nonincreasing, aligning with the problem~\eqref{eq.ml_problem_para}. We also empirically observe that OffSBL iterates also converge.

\subsection{Algorithm Complexity}
For simplicity and interpretability, we assume $M_i=\mathcal{O}(M)$, $N_i=\mathcal{O}(N)$, and $I<M<N$, where $N_i$ is the number of grids adopted in the BEM of~\eqref{eq.deco_linear_inversion}. The time complexity of dSBL with HOSVD is $\mathcal{O}\left(R_\mathrm{EM}N^2MI + IM^{I+1}\right)$ while low-complexity approximation has $\mathcal{O}\left(R_\mathrm{EM}N^2MI + M^{I+1}\right)$. Here, $R_\mathrm{EM}$ denotes the number of EM iterations. The difference between HOSVD and low complexity approximation-based schemes is roughly of order $I$. Meanwhile, both have similar space complexity, i.e., $\mathcal{O}(M^I+MN+N^2)$. Also, all sparse recovery subproblems are independent of each other and can thus be solved in parallel. In that case, the time complexity changes to $\mathcal{O}(M^{I+1}+R_\mathrm{EM}N^2M)$ for low-complexity approximation based and $\mathcal{O}(IM^{I+1}+R_\mathrm{EM}N^2M)$ for HOSVD based, while the space complexity becomes $\mathcal{O}(M^I+IMN+IN^2)$. For comparison, the time complexity of AM- and SVD-KroSBL is $\mathcal{O}\big(R_\mathrm{EM}( R_\mathrm{AM} IN^I +(MN)^I )\big)$ and $\mathcal{O}\big(R_\mathrm{EM}( N^{I+1} +(MN)^I )\big)$, respectively, and the space complexity is $\mathcal{O}((MN)^I)$ for both. Here, $R_\mathrm{AM}$ denotes the number of AM iterations in AM-KroSBL. Therefore, both the time and space complexities of our algorithm are several orders less than the state-of-the-art KroSBL methods.

\subsection{Extensions to Similar Structures}\label{sec.extension}
We reiterate that the decomposition and the OffSBL algorithms we presented are general algorithmic techniques and can also be applied independently. For instance, for estimation tasks involving Kronecker-structured signals, if there is no grid mismatch and the parameters $\bar{\bm \psi}_i$ lie on a discrete set, one can combine the decomposition algorithm with any sparse recovery algorithm. Similarly, OffSBL is a stand-alone off-grid algorithm for conventional linear inversion problems (i.e., $I=1$). Moreover, these techniques can be extended to other parameter estimation models, as discussed next.
\subsubsection{Superposition of Kronecker-structured Data}
In several wideband orthogonal frequency division multiplexing multiple input multiple output (MIMO) systems~\cite{zhou2017low,araujo2019tensor,wang2024tensor}, measurements $\bar{\bm y}$ takes the form
\begin{equation}\label{eq.superposition_Kro}
    \bar{\bm y} = \sum_{u=1}^U \left(\otimes_{i=1}^I\bar{\bm H}_{u,i,\bar{\bm \psi}_{u,i}} \bar{\bm x}_{u,i} \right)+\bar{\bm n}.
\end{equation}
where the special case of $U=1$ reduces to \eqref{eq.Kro_generate_data}. Here, we can rewrite~\eqref{eq.superposition_Kro} using the tensor form as
\begin{equation*}%\label{eq.superposition_Kro_tensor}
    \mathcal{\bm Y} = \sum_{u=1}^U \circ_{i=1}^I\left(\bar{\bm H}_{u,i,\bar{\bm \psi}_{u,i}} \bar{\bm x}_{u,i} \right)+\bar{\mathcal{\bm N}}.
\end{equation*}
Thus, each factor $\bar{\bm H}_{u,i,\bar{\bm \psi}_{u,i}} \bar{\bm x}_{u,i}$ for $u=[U]$ and $i\in[I]$ can be obtained by tensor canonical polyadic decomposition under mild conditions for uniqueness of the decomposition \cite{cichocki2015tensor}, followed by OffSBL for unknown parameter estimation.

\subsubsection{Non-Kronecker-structured Sparse Vector}
Some applications~\cite{caiafa2012block,chang2021sparse,he2023bayesian} lead to the measurement model
\begin{equation*}
    \bar{\bm y} =\left(\otimes_{i=1}^I\bar{\bm H}_{i,\bar{\bm \psi}_i}\right)\bar{\bm x}+\bar{\bm n},
\end{equation*}
where the coefficient vector $\bar{\bm x}$ is not Kronecker-structured. 
Here, direct decomposition of $\bar{\bm y}$ cannot be applied, but we can use the Kronecker-structured dictionary in multidimensional BEMs as $\bm H=\otimes_{i=1}^I \bm H_i\left( \bm \psi_{(i)} \right)$. Here, vector $\bm \psi_{(i)}$ collects all the grid points in the $i$th BEM dictionary matrix $\bm H_i$. Then, we arrive at the sparse vector problem
\begin{equation}\label{eq:Non-kro_OffSBL}
    \bar{\bm y} =\otimes_{i=1}^I \bm H_i\left( \bm \psi_{(i)} \right)\bm x+\bar{\bm n},
\end{equation}
which can be solved using the OffSBL algorithm. Specifically, we adopt a fictitious Gaussian prior distribution with covariance matrix $\bm \Gamma = \diag(\bm \gamma) \in \mathbb{R}^{\bar{N}\times \bar{N}}$ on $\bm x$. Then the estimates of $\bm \gamma$, $\{\bm \psi_{(i)}\}_{i=1}^I$ and $\sigma^2$ are determined by the type II ML with EM algorithm. The hyperparameter $\bm \gamma$ and the noise variance $\sigma^2$ can be similarly obtained using~\eqref{eq.qfunc} and~\eqref{eq.noise_est}. We can exploit the Kronecker structure in~\eqref{eq:Non-kro_OffSBL} to simplify the alternating minimization of OffSBL for $\{\bm \psi_{(i)}\}_{i=1}^I$. The EM update step for $\{\bm \psi_{(i)}\}_{i=1}^I$ is equivalent to minimizing 
\begin{equation*}%\label{eq.cost_kro_dict}
    g\left( \{\bm \psi_{(i)}\}_{i=1}^I \right)=\trace\left(\bm H \bm \Sigma \bm H^\mathsf{H}\right)-2\mathsf{Re}\left\{\trace\left(\bm M\bm H\right)\right\}.
\end{equation*}
Let $\psi_{i^*,n^*}$ be the $n^*$th grid point of the $i^*$th BEM dictionary matrix $\bm H_{i^*}$. Then, the alternating minimization step optimizes $\psi_{i^*,n^*}$ with other $\psi_{i,n}$'s being fixed, as in OffSBL. To this end, for any $i^*$, we can reorder the Kronecker product as 
\begin{equation*}
    \bm H = \otimes_{i=1}^I \bm H_i = \bm P_{i^*}(\bm H_{i^*} \otimes\bm S_{i^*})\bm Q_{i^*},
\end{equation*}
where $\bm S_{i^*} = (\otimes_{i=i^*+1}^I\bm H_i) \otimes (\otimes_{i=1}^{i^*-1}\bm H_i)$ is independent of $\bm H_{i^*}$ and $\bm P_{i^*}$ and $\bm Q_{i^*}$ are the corresponding permutation matrices~\cite{van2000ubiquitous}. Thus, the update step for $\psi_{i^*,n^*}$ minimizes
\begin{align*}
    f_{i^*,n^*}(\psi_{i^*,n^*}) &=\trace \left( \left[\left( \bm H_{i^*}^\mathsf{H} \bm H_{i^*} \right) \otimes \left( \bm S_{i^*}^\mathsf{H} \bm S_{i^*} \right)\right] \bm Q_{i^*} \bm \Sigma\bm Q_{i^*}^\mathsf{H}\right)\\
    &\;\;+\trace([\bm H_{i^*} \otimes\bm S_{i^*}]\bm Q_{i^*}\bm M\bm P_{i^*})\\
    &= 2 \mathsf{Re}\{ \bm v_{i^*,n^*}^\mathsf{H} \!\bm h_{i^*,n^*} \} \!+\! c_{i^*,n^*}\|\bm h_{i^*,n^*}\|_2^2\!+\!\rho_{i_*,n^*}\!,
\end{align*}
% the first term of \eqref{eq.cost_kro_dict} can be further expanded as
% \begin{multline}\label{eq.expand_qua}
%     \trace \left( \bm H \bm \Sigma \bm H^\mathsf{H}\right)
%         = \trace \left( (\bm H_{i^*}^\mathsf{H} \otimes\bm S_{i^*}^\mathsf{H}) (\bm H_{i^*} \otimes\bm S_{i^*})\bm Q_{i^*}  \bm \Sigma \bm Q_{i^*}^\mathsf{H} \right) \\
%         = \trace \left( \left( \bm H_{i^*}^\mathsf{H} \bm H_{i^*} \right) \otimes \left( \bm S_{i^*}^\mathsf{H} \bm S_{i^*} \right) \bm Q_{i^*} \bm \Sigma\bm Q_{i^*}^\mathsf{H}\right).
% \end{multline}
% Similarly, the second term of \eqref{eq.cost_kro_dict} is
% \begin{multline}\label{eq.expand_linear}
%         \trace (\bm M\bm H) = \trace((\bm H_{i^*} \otimes\bm S_{i^*})\bm Q_{i^*}\bm M\bm P_{i^*})
%         \\= \trace \left( \left( \bm H_{i^*} \otimes \left(\otimes_{i=i^*+1}^I\bm H_i\right)\otimes\left(\otimes_{i=1}^{i^*-1}\bm H_i \right)\right) \bm P_r\bm M\bm P_l \right),
% \end{multline}
% where $\bm P_r$ and $\bm P_l$ are permutation matrices. Based on~\eqref{eq.expand_qua}, ~\eqref{eq.expand_linear} and fixing $g\left( \{\bm \psi_i\}_{i=1}^I \right)$ but $\psi_{i^*,n^*}$,~\eqref{eq.cost_kro_dict} can be 
% similarly translated into a quadratic cost function $g(\psi_{i^*,n^*})$, given~by
% \begin{equation}\label{prob.single_op_kro}
%         g(\psi_{i^*,n^*}) = 2 \mathsf{Re}\{ \bm v_{i^*,n^*}^\mathsf{H} \bm h_{i^*,n^*} \} + c\|\bm h_{i^*,n^*}\|_2^2,
%     \end{equation}
where $\bm v_{i^*,n^*}$, $c_{i^*,n^*}$ and $\rho_{i_*,n^*}$ are independent of $\psi_{i^*,n^*}$. Thus, $f_{i^*,n^*}(\psi_{i^*,n^*})$ can be efficiently minimized with respect to $\psi_{i^*,n^*}$ using a line search. Also, alternating minimization sequentially updates $\psi_{i^*,n^*}$ for different values of $i^*$ and $n^*$, unlike dSBL where parallel updates are possible. Consequently, OffSBL incurs a higher computational cost here.

\section{Application: Channel Estimation for IRS-aided MIMO System}
\label{sec.channelmodel}

% \begin{figure}[t!]
% \centering
% \includegraphics[width=0.45\linewidth]{figure/IRS_two_scatters.pdf}
% \caption{An illustration of AoAs and AoDs in the IRS-aided uplink channel.}
% \label{fig:irs_sys}
% \end{figure}

In this section, we explore the application of our algorithm to cascaded channel estimation in an IRS-assisted MIMO system.
We consider an uplink MIMO millimeter-wave system with a transmitter MS with $T$ antennas, a receiver BS with $R$ antennas, and a uniform linear array IRS with $L$ elements. Let $\bm \Phi_\mathrm{MS}\in \mathbb{C}^{L\times T}$ and $\bm \Phi_\mathrm{BS} \in \mathbb{C}^{R\times L}$ denote the geometric narrowband MS-IRS and IRS-BS channel, respectively,
\begin{align}
\label{eq.channelmodel1}
\bm \Phi_\mathrm{MS}&=\sum_{p=1}^{P_{\mathrm{MS}}}\sqrt{\frac{LT}{P_{\mathrm{MS}}}}\beta_{{\mathrm{MS}},p}\bm a_L(\phi_{{\mathrm{MS}},p})\bm a_{T}(\alpha_{{\mathrm{MS}}})^\mathsf{H},\\\label{eq.channelmodel2}
\bm \Phi_\mathrm{BS}&=\sum_{p=1}^{P_{\mathrm{BS}}}\sqrt{\frac{RL}{P_{\mathrm{BS}}}}\beta_{{\mathrm{BS},p}}\bm a_{R}(\alpha_{{\mathrm{BS},p}})\bm a_L(\phi_{{\mathrm{BS}}})^\mathsf{H},
\end{align}    
where we define the steering vector $\bm a_Q(\psi)\in\mathbb{C}^{Q}$ for an integer $Q$ and angle $\psi$ as 
\begin{equation*}%\label{eq.steering}
\bm a_Q(\psi) = 1/\sqrt{Q}[1,e^{j\frac{2\pi \Delta}{\eta}\cos{\psi}},\cdots,e^{j \frac{2\pi \Delta}{\eta} (Q-1)  \cos{\psi}}]^\mathsf{T}    
\end{equation*}
Here, $\Delta$ is the distance between two adjacent elements, and $\eta$ is the wavelength. We denote the number of rays in the scatter as $P_{\mathrm{MS}}$ and $P_{\mathrm{BS}}$. The angles $\phi_{{\mathrm{MS}},p}$, $\alpha_{\mathrm{MS}}$, $\alpha_{{\mathrm{BS},p}}$, and $\phi_{{\mathrm{BS}}}$ denote the $p$th AoA of the IRS, and the AoD of the MS, the $p$th AoA of the BS, and the AoD of the IRS, respectively (see \cite[Fig. 1]{he2022structure}). Then, the cascaded MS-IRS-BS channel can be expressed as $\bm \Phi_\mathrm{BS}\diag (\bm \omega)\bm \Phi_\mathrm{MS}$ for any IRS configuration $\bm\omega \in \mathbb{C}^{L }$. The $l$th entry of $\bm\omega$ represents the gain and phase shift due to the $l$th IRS element. Our goal is to estimate the cascaded channel $\bm \Phi_\mathrm{BS}\diag (\bm \omega)\bm \Phi_\mathrm{MS}$, which is a function of angles $\phi_{{\mathrm{MS}},p}$, $\alpha_{\mathrm{MS}}$, $\alpha_{{\mathrm{BS},p}}$, and $\phi_{{\mathrm{BS}}}$, for a given $\bm \omega$.

We estimate the channel using pilot data transmitted over $K$ time slots. We choose $K_{\mathrm{I}}$ IRS configurations, and for each configuration, transmit pilot data $\bm G \in \mathbb{C}^{T\times K_{\mathrm{P}}}$ over $K_{\mathrm{P}}$ time slots, where $K=K_{\mathrm{I}}K_{\mathrm{P}}$. For the $k$th configuration $\bm \omega_k$, the received signal $\bm Y_k=\bm \Phi_\mathrm{BS}\diag (\bm \omega_k)\bm \Phi_\mathrm{MS} \bm G+\bm N_k\in \mathbb{C}^{R \times K_{\mathrm{P}}}$
%\begin{equation}\label{eq.basicdatamodel}
%$\bm Y_k=\bm \Phi_\mathrm{BS}\diag (\bm \omega_k)\bm \Phi_\mathrm{MS} \bm G+\bm N_k$,
%\end{equation}
where $\bm N_k \in \mathbb{C}^{R\times K_{\mathrm{P}}}$ is the noise. Using~\eqref{eq.channelmodel1} and~\eqref{eq.channelmodel2}, we~get
\begin{multline*}
    \bm Y_k=\zeta\sqrt{L} \bm A_{R,\mathrm{BS}} \bm \beta_{{\mathrm{BS}}}\bm a_L(\phi_{{\mathrm{BS}}})^\mathsf{H}\diag (\bm \omega_k)\\\times
    \bm A_{L,\mathrm{MS}}\bm \beta_{\mathrm{MS}}\bm a_{T}(\alpha_{{\mathrm{MS}}})^\mathsf{H} \bm G+\bm N_k,
\end{multline*}
where $\bm{A}_{R,\mathrm{BS}} \in \mathbb{C}^{R \times P_{\mathrm{BS}}} $ and $\bm{\beta}_{\mathrm{BS}} \in \mathbb{C}^{P_{\mathrm{BS}}} $ have $\bm{a}_{R}(\alpha_{\mathrm{BS},p}) $ and $\beta_{\mathrm{BS},p} $ as their $ p $th column and entry, respectively. Similarly,  $\bm{A}_{L,\mathrm{MS}} \in \mathbb{C}^{R \times P_{\mathrm{MS}}} $ and $\bm{\beta}_{\mathrm{MS}} \in \mathbb{C}^{P_{\mathrm{MS}}} $ have $\bm{a}_{L}(\alpha_{\mathrm{MS},p}) $ and $\beta_{\mathrm{MS},p} $ as their $ p $th column and entry, respectively. Also, $\zeta \coloneq \sqrt{\frac{LRT}{P_{\mathrm{MS}}P_{\mathrm{BS}}}} $. Vectorizing $\bm Y_k$ and using the properties of the Khatri-Rao product~\cite[Lemma A1]{rao1970estimation} leads to (see \cite{he2022structure} for detailed algebraic simplifications)
\begin{align*}
    \bm y_k &= \zeta\sqrt{L} \left(\bm A_{L,\mathrm{MS}}\bm \beta_{\mathrm{MS}}\bm a_{T}(\alpha_{{\mathrm{MS}}})^\mathsf{H} \bm G \right)^\mathsf{T}\\&\;\;\;\;\odot\left( \bm A_{R,\mathrm{BS}} \bm \beta_{{\mathrm{BS}}}\bm a_L(\phi_{{\mathrm{BS}}})^\mathsf{H}  \right)\bm \omega_k + \bm n_k\\
    &=\left[\bm \omega_k^\mathsf{T}\bm A_{L,\mathrm{I}}\bm \beta_{\mathrm{MS}}\right]
         \left[\zeta\bm G^\mathsf{T}\bm a_{T}^*(\alpha_{{\mathrm{MS}}}) \right]
        \otimes \left[\bm A_{R,\mathrm{BS}} \bm \beta_{{\mathrm{BS}}}\right]+\bm n_k,
\end{align*}
where $\bm{A}_{L,\mathrm{I}} \in \mathbb{C}^{L \times P{\mathrm{MS}}}$ whose $p$th column is $\bm a_{L}(\phi_{\mathrm{MS},p}-\phi_\mathrm{BS})$, $\odot$ is the Khatri-Rao product, and $(\cdot)^*$ is conjugate. Here, the channel is given by
\begin{multline}\label{eq:channel_repre}
    \mathrm{vec}(\bm \Phi_\mathrm{BS}\diag (\bm \omega)\bm \Phi_\mathrm{MS}) \\= \left[\bm \omega^\mathsf{T}\bm A_{L,\mathrm{I}}\bm \beta_{\mathrm{MS}}\zeta\bm a_{T}^*(\alpha_{{\mathrm{MS}}})\right]
        \otimes \left[\bm A_{R,\mathrm{BS}} \bm \beta_{{\mathrm{BS}}}\right].
\end{multline}
Combining the received signals obtained for the $K_{\mathrm{I}}$ configurations, 
% \begin{multline}
%     \bar{\bm Y} = \zeta \left(\bm A_{\mathrm{L,MS}}\bm \beta_{\mathrm{MS}}\bm a_{T}(\alpha_{{\mathrm{MS}}})^\mathsf{H} \bm G \right)^\mathsf{T}\\\odot\left( \beta_{{\mathrm{BS}}}\bm a_{R}(\alpha_{{\mathrm{BS}}})\bm a_L(\phi_{{\mathrm{BS}}})^\mathsf{H}  \right)\bm \Omega+\bar{\bm N},
% \end{multline}
we obtain $\bar{\bm y}\in\mathbb{C}^{RK}$
% \begin{multline}\label{eq.ce_data_model}
%         \bar{\bm y}=\left[\bm \Omega^\mathsf{T}\left( \bm A_{L,\mathrm{MS}}^\mathsf{T} \odot \bm a_L(\phi_{{\mathrm{BS}}})^\mathsf{H} \right)^\mathsf{T}\bm \beta_{\mathrm{MS}}\right]
%         \\\otimes \left[ \zeta\bm G^\mathsf{T}\bm a_{T}^*(\alpha_{{\mathrm{MS}}}) \right]
%         \otimes \left[\bm A_{R,\mathrm{BS}} \bm \beta_{{\mathrm{BS}}}\right]+\bar{\bm n}.
% \end{multline}
\begin{equation}\label{eq.ce_data_model}
        \bar{\bm y}=\left[\bm \Omega^\mathsf{T}\bm A_{L,\mathrm{I}}\bm \beta_{\mathrm{MS}}\right]
        \otimes \left[ \zeta\bm G^\mathsf{T}\bm a_{T}^*(\alpha_{{\mathrm{MS}}}) \right]
        \otimes \left[\bm A_{R,\mathrm{BS}} \bm \beta_{{\mathrm{BS}}}\right]+\bar{\bm n},
\end{equation}
% \begin{equation}
% \begin{aligned}
%         \bar{\bm y}&=\bigg(\bm \Omega^\mathsf{T}\left( \bm A_{\mathrm{L,MS}}^\mathsf{T} \odot \bm a_L(\phi_{{\mathrm{BS}}})^\mathsf{H} \right)^\mathsf{T}
%     \\&\otimes \left( \bm G^\mathsf{T}\bm a_{T}^*(\alpha_{{\mathrm{MS}}}) \right)\otimes \bm a_{R}(\alpha_{{\mathrm{BS}}})\bigg) \vect (\zeta\beta_{{\mathrm{BS}}}\bm \beta_{\mathrm{MS}})+\bar{\bm n}
%     \\&=\left(  \bm \Omega^\mathsf{T}\left( \bm A_{\mathrm{L,MS}}^\mathsf{T} \odot \bm a_L(\phi_{{\mathrm{BS}}})^\mathsf{H} \right)^\mathsf{T}\bm \beta_{\mathrm{MS}}\right)\otimes \left( \zeta\bm G^\mathsf{T}\bm a_{T}^*(\alpha_{{\mathrm{MS}}}) \right) 
%     \\&\otimes \beta_{{\mathrm{BS}}} \bm a_{R}(\alpha_{{\mathrm{BS}}})+\bar{\bm n}.
% \end{aligned}
% \end{equation}
where $\bm{\Omega} \in \mathbb{C}^{L \times K_\mathrm{I}}$ with the $k$th column as $\bm\omega_k$. Therefore, from \eqref{eq:channel_repre}, the channel estimation problem reduces to recovering $\bm A_{L,\mathrm{I}}\bm \beta_{\mathrm{MS}}$, $\bm a_{T}(\alpha_{{\mathrm{MS}}})$, and $\bm A_{R,\mathrm{BS}} \bm \beta_{{\mathrm{BS}}}$ from $\bar{\bm y}$ up to a scaling factor, given that  $\bm G$ and $\bm \Omega$ are known. 

Comparing~\eqref{eq.ce_data_model} with~\eqref{eq.Kro_generate_data}, we see the signal model here is the Kronecker product of three terms, i.e., $I=3$. The unknown parameters are AoA or AoD given by $\bar{\bm \psi}_1\in\mathbb{R}^{P_{\mathrm{MS}}}$, $\bar{\bm \psi}_2= \alpha_\mathrm{MS}\in\mathbb{R}$, and $\bar{\bm \psi}_3\in\mathbb{R}^{P_{\mathrm{BS}}}$ where the $p$th entry of $\bar{\bm \psi}_1$ and $\bar{\bm \psi}_3$ are $\phi_{\mathrm{MS},p} - \phi_\mathrm{BS}$ and $\alpha_{\mathrm{BS},p}$, respectively. Also, $\psi_{i,\mathrm{l}}=-\pi$ and $\psi_{i,\mathrm{r}}=\pi$ for all values of $i$. The basis functions are related to the steering vectors as 
\begin{equation*}%\label{eq.col_func}
    \bm h_1(\psi) = \bm \Omega^\mathsf{T}\bm a_L(\psi),\ \bm h_2(\psi) = \bm G^\mathsf{T}\bm a_T(\psi),\ \bm h_3(\psi) = \bm a_R(\psi).
\end{equation*}
Correspondingly, we have $\bar{\bm x}_1 = \bm \beta_{\mathrm{MS}}$, $\bar{\bm x}_2 = \zeta$, and $\bar{\bm x}_3 = \bm \beta_{\mathrm{BS}}$. The channel estimation problem is now translated into estimating $\{\bar{\bm \psi}_i\}_{i=1}^3$ and $\{\bar{\bm x}_i\}_{i=1}^3$ from the noisy measurement $\bar{\bm y}$, where our dSBL (Algorithm~\ref{al.offSBL}) can be applied.

Using the decomposition step (Line~2) of Algorithm~\ref{al.offSBL}, we first decompose $\bar{\bm y}$ into three vectors, $\hat{\bm y}_1\in \mathbb{C}^{K_{\mathrm{I}}} $, $\hat{\bm y}_2\in \mathbb{C}^{K_\mathrm{P}} $, and $\hat{\bm y}_3\in \mathbb{C}^{R} $, corresponding to the three terms in the Kronecker product in \eqref{eq.ce_data_model}. 
% as
% \begin{equation}\label{eq.after_deco}
% \begin{aligned}
%     \bm y_1 &= \bm \Omega^\mathsf{T}\left( \bm A_{\mathrm{L,MS}}^\mathsf{T} \odot \bm a_L(\phi_{{\mathrm{BS}}})^\mathsf{H} \right)^\mathsf{T}\bm \beta_{\mathrm{MS}}+\bm n_1 \in \mathbb{C}^{M_1},\\
%     \bm y_2 &= \bm G^\mathsf{T}\bm a_{T}^*(\alpha_{{\mathrm{MS}}}) \zeta+\bm n_2\in \mathbb{C}^{M_2},\\
%     \bm y_3 &= \bm A_{R,\mathrm{BS}} \bm \beta_{{\mathrm{BS}}}+\bm n_3\in \mathbb{C}^{M_3},
% \end{aligned}
% \end{equation}
% where $M_1 = K_{\mathrm{I}}$, $M_2=K_\mathrm{P}$, $M_3 = R$, $\bar{M}=RK$, and we set ambiguity factors to one. 
We then use the dictionary for the steering vectors for a given integer $Q$ as
\begin{equation}\label{eq.dict_form}
\bm A_{Q}(\bm \psi) = \begin{bmatrix}
\bm a_{Q}(\psi_1) & \bm a_{Q}(\psi_2)&\ldots&\bm a_{Q}(\psi_N)
\end{bmatrix}\in \mathbb{C}^{Q\times N},
\end{equation}
where $\bm \psi$ captures the unknown angles.
This formulation leads to the following three problems similar to \eqref{eq.problem_basic_bem},
% \begin{align}
% \bm y_1 &= \bm H_1 \bm x_1 + {\bm n_1},\label{eq.sys1}\tag{P1}\\
% \bm y_2 &= \bm H_2 \bm x_2 + {\bm n_2},\label{eq.sys2}\tag{P2}\\
% \bm y_3 &= \bm H_3 \bm x_3 + {\bm n_3},\label{eq.sys3}\tag{P3}
% \end{align}
% with $\bm H_{\mathrm{L}}\coloneq\bm \Omega^\mathsf{T}\bm A_L \in\mathbb{C}^{K_{\mathrm{I}}\times N}$, $\bm H_{\mathrm{T}}\coloneq\bm G^\mathsf{T} \bm A_{T}^*$, $\bm H_{\mathrm{R}}\coloneq\bm A_{R}$, $I=3$, $M_1 = K_{\mathrm{I}}$, $M_2=K_\mathrm{P}$, $M_3 = R$, and $\bar{M}=RK$. We note that channel reconstruction only requires $\bm x_{\mathrm{L}} \otimes \bm x_{\mathrm{T}}^* \otimes \bm x_{\mathrm{R}}$, which is thus not affected by scaling ambiguity.
% Problem~\eqref{eq.sys1},~\eqref{eq.sys2}, and~\eqref{eq.sys3} can be unified as
\begin{equation}\label{eq.simplified_bem}
    \hat{\bm y}_i = \bm H_i (\bm \psi_{(i)})\bm x_i+\bm n_i ,\ i=1,2,3,
\end{equation}
where $\bm H_1(\bm \psi_{(1)})=\bm \Omega^\mathsf{T}\bm A_{L}(\bm \psi_{(1)})$, $\bm H_2(\bm \psi_{(2)})=\bm  G^\mathsf{T}\bm A_{T}(\bm \psi_{(2)})$, and $\bm H_3(\bm \psi_{(3)})=\bm A_{R}(\bm \psi_{(3)})$.
% \begin{equation}\label{eq.col_func}
%     \bm h_1(\cdot) = \bm \Omega^\mathsf{T}\bm a_L(\cdot),\ \bm h_2(\cdot) = \bm G^\mathsf{T}\bm a_T(\cdot),\ \bm h_3(\cdot) = \bm a_R(\cdot).
% \end{equation}
We solve them via OffSBL, initializing $\bm \psi_{(i)}$ by sampling the angular domain using $N$ grid angles $\{\theta_n\}_{n=1}^N$ such that $\cos(\theta_n)= 2(n-1)/N-1$. OffSBL provides estimates $(\hat{\bm\psi}_{(i)},\hat{\bm x}_i)$ of $(\bar{\bm\psi}_i,\bar{\bm x}_i)$, for $i=1,2,3$. Finally, using \eqref{eq:channel_repre} we compute the channel estimate for a given configuration $\bm \omega$ as 
\begin{equation*}\left[\bm \omega^\mathsf{T}\bm A_{L}(\hat{\bm \psi}_{(1)})\hat{\bm x}_1\bm A_{T}(\hat{\bm \psi}_{(2)})\hat{\bm x}_2\right]
        \otimes \left[\bm A_{R}(\hat{\bm \psi}_{(3)})\hat{\bm x}_3\right].
        \end{equation*} Here, the channel estimate is not affected by scaling ambiguity. 

We reiterate that OffSBL is a standalone off-grid algorithm applicable to various linear inversion problems (i.e., $I=1$). One notable example is direction-of-arrival estimation, where $K$ far-field narrowband signals impinge on a uniform array with $L$ elements ($K < M$), resulting in the received signal given by
\begin{equation*}
    \hat{\bm y} = \bm A_{L}(\bm\psi) \bm x + {\bm n},
\end{equation*}
where the BEM dictionary $\bm A_{L}(\bm \psi)$ defined in \eqref{eq.dict_form} has steering vectors as its columns, $\bm \psi$ captures the AoAs, and $\bm x$ is the sparse vector whose support corresponds to the true AoAs. This formulation is similar to $i=3$ in \eqref{eq.simplified_bem} for the IRS-aided channel estimation problem. In several other applications, the BEM dictionary takes the form $\bm B\bm A_{L}(\bm\psi)$. Specifically, for $i=1$ and $i=2$ in \eqref{eq.simplified_bem}, $\bm B$ corresponds to $\bm \Omega^\mathsf{T}$ and $\bm G^\mathsf{T}$, respectively. In other cases, $\bm B$ 
        % We see that $i=3$ is equivalent to the direction of arrival estimation, while the dictionaries in $i=1,2$ are functions of steering vectors. Since our OffSBL (Algorithm \ref{al.offSBL}) does not assume any structure on $\bm h_i$, we apply it to reconstruct the sparse coefficients $\bm x_i$ (path gains) and to estimate the unknown parameters (angles) for $i=1,2,3$. We note that the column functions in~\eqref{eq.col_func} also exist in other problems. In different scenarios, the matrix $\bm \Omega$ (or $\bm G$) 
can represent beamformers \cite{wang2024tensor,schroeder2024low,zhou2017low}, combiners \cite{zhou2017low}, IRS configurations \cite{wang2024intelligent},  pilot data \cite{wang2024tensor}, or random matrices~\cite{wang2009direction}. 

\section{Performance Evaluation}\label{sec.simu}
We present numerical results to compare our algorithm with the state-of-the-art methods. We present three sets of results. The first two demonstrate the decomposition step and OffSBL for parameter estimation. The third shows the combined results for IRS-aided channel estimation.
Our code is available \href{https://github.com/YanbinHe/JournalDecomOffGrid}{here}.

\subsection{Decomposition-based Sparse Vector Recovery}
In this section, we highlight the advantages of the decomposition step in reducing computational complexity and enhancing the denoising effect. We focus on recovering the Kronecker-structured sparse vector~\eqref{eq.vect_kro} using a multidimensional BEM~\eqref{eq.problem_basic} in the on-grid setting, without requiring the OffSBL algorithm. By combining the decomposition step with  SBL, we demonstrate the benefits of this approach. We compare our method's performance with other methods that do not use decomposition, such as classical SBL~\cite{wipf2004sparse}, classical OMP, AM-KroSBL, and SVD-KroSBL~\cite{he2022structure}. Specifically, AM- and SVD-KroSBL only consider the Kronecker-structured support of the sparse vector and do not exploit the Kronecker structure in the nonzero entries as in~\eqref{eq.vect_kro}.

We set $I=3$ with $M_i=M$, and $N_i=12$ for $i=1,2,3$ in~\eqref{eq.problem_basic}~and~\eqref{eq.vect_kro}. So, we have $\bm H = \otimes_{i=1}^3\bm H_i$ and $\bm x = \otimes_{i=1}^3\bm x_i$ with $\bm x_i\in\mathbb{R}^{12}$. The columns of $\bm H_i \in \mathbb{C}^{M\times 12}$ for $i=1,2,3$ are the steering vectors evaluated by the grids $\{\theta_n\}_{n=1}^{12}$ defined in Sec. \ref{sec.channelmodel}. There are four nonzeros in each $\bm x_i$, whose positions are uniformly chosen from the grids and amplitudes are uniformly drawn from $[0.5,1.5]$. Here, measurement level $M$ is set to be $\{8,9,10\}$, labeled as \texttt{Low}, \texttt{Medium}, and \texttt{High} measurement case, controlling the number of measurements $\bar{M}=M^3$ and the undersampling ratio $M^3/N^3$. We adopt the additive white Gaussian noise with zero mean whose variance is determined by $\text{SNR~(dB)} = 10\log_{10}\mathbb{E}\{\|\bm H\bm x\|_2^2/\|\bm n\|_2^2\}$ of $\{5,10,15,20,25,30\}$. Three metrics are considered for performance evaluation: normalized mean squared error (NMSE), support recovery rate (SRR), and run time. Here, we define 
$$
    \text{NMSE} = \mathbb{E}\left\{\frac{\|\bm x-\hat{\bm x}\|_2^2}{\|\bm x\|_2^2}\right\},
    \text{SRR} = \frac{|\supp(\hat{\bm x}) \cap \supp(\bm x) |}{|\supp(\hat{\bm x}) \cup \supp(\bm x) |}\label{eq.def_srr},
$$
where ${\bm x}$ is the ground truth and $\hat{\bm x}$ is the estimated vector. We limit the number of iterations for the SBL-based methods (dSBL, cSBL, AM-KroSBL, and SVD-KroSBL) to 200 and prune small entries in hyperparameters for faster convergence.
% The results are summarized in Fig.~\ref{fig.sparse_re} and Tables~\ref{tab.denoise} and \ref{tab.time_decom}. 
\begin{figure}
\centering
  \includegraphics[width=0.49\linewidth]{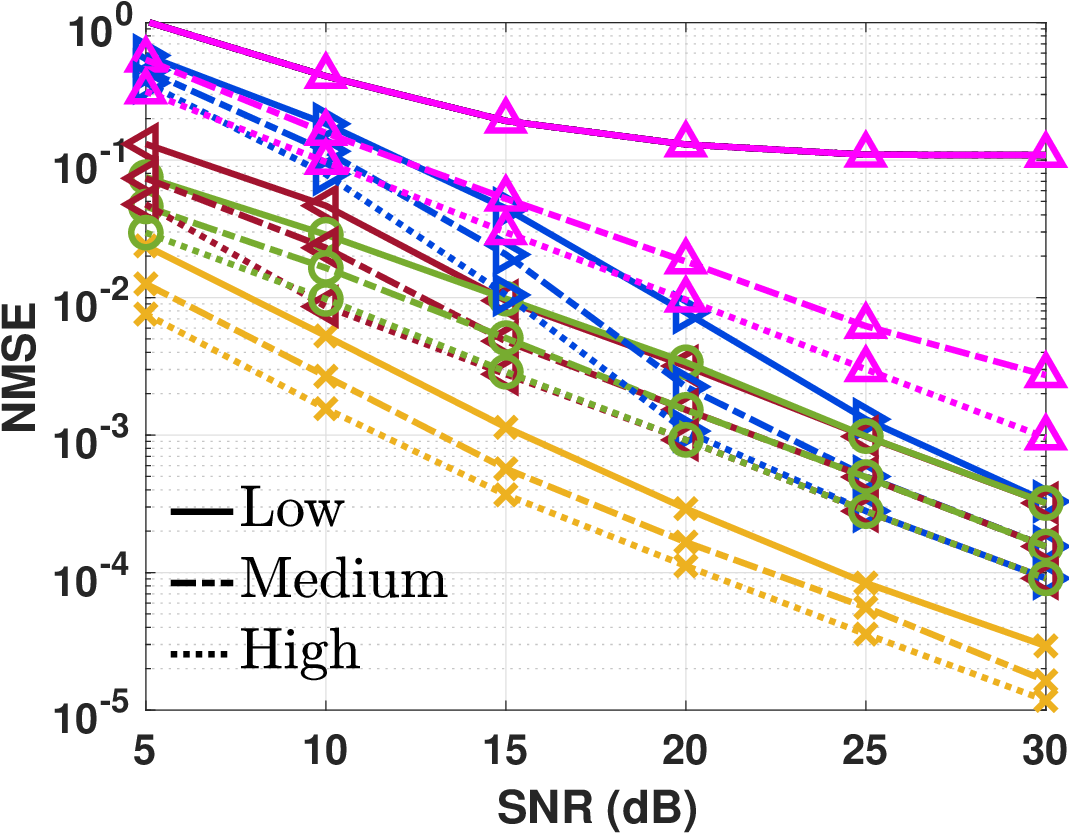}
  \includegraphics[width=0.49\linewidth]{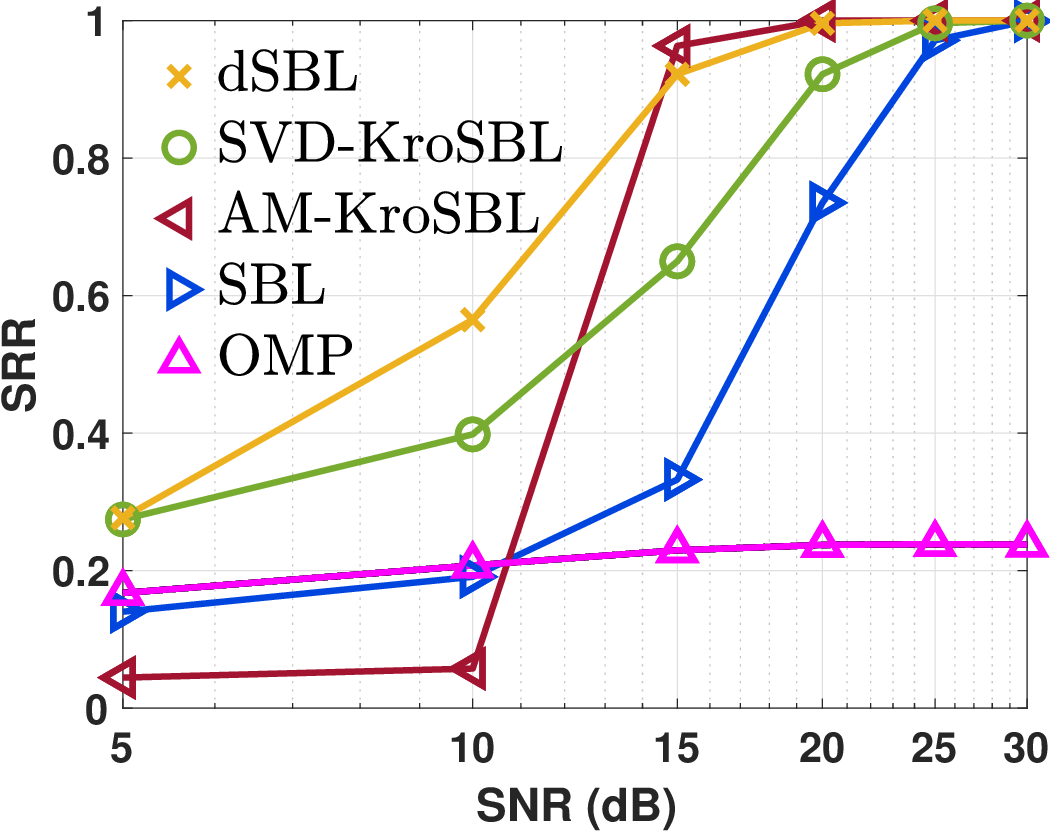}
  \caption{NMSE and SRR of different algorithms as functions of SNR. For the SRR curves, we choose $M=8$ to avoid clutter.}
  \label{fig.sparse_re}
\end{figure}

The denoising effect of the decomposition step is shown in Table~\ref{tab.denoise}. Here, we compare the noise levels of $i)$ the original noisy signal $\bm y$ in~\eqref{eq.problem_basic}, $ii)$ the signal after decomposition $\hat{\bm y}\coloneq\otimes_{i=1}^3 \hat{\bm y}_i$ where $\hat{\bm y}_i$'s are the solution to the problem~\eqref{prob.bidecom}, and $iii)$ the result of \eqref{eq.appro_denoise_hosvd}. It can be seen that the noise level is significantly reduced after decomposition. It also closely matches the result in \eqref{eq.appro_denoise_hosvd}, validating our claim on the denoising effect discussed in Sec.~\ref{sec:estimation_acc}. 

Fig.~\ref{fig.sparse_re} shows that with higher SNR and more measurements, all algorithms yield better NMSE and SRR, as expected. Our dSBL algorithm outperforms other methods in NMSE and has the best SRR performance in most cases, demonstrating the efficacy of the decomposition idea. In contrast to the SVD-KroSBL algorithm that uses Kronecker-structured support, dSBL achieves superior NMSE by using the additional Kronecker structure in nonzero entries explicitly enforced via~\eqref{eq.vect_kro} through the decomposition step. The relatively lower performance of AM-KroSBL is attributed to its slow convergence, given that we fix the number of EM iterations, as pointed out in \cite{he2023bayesian}. The lower SRR and NMSE observed in the low SNR regime are due to small nonzero values in the estimate at locations where the ground truth is zero. 

Finally, Table~\ref{tab.time_decom} demonstrates that dSBL requires two-order less run time than the other competing algorithms, corroborating the computational advantage of our decomposition.

\begin{table}[t]
\centering
\scriptsize
\caption{Illustration of denoising with $M=10$, using the original noisy signal $\bm y$, reconstructed signal $\hat{\bm y}=\otimes_{i=1}^3 \hat{\bm y}_i$ after the decomposition step, and ground truth $\bm y_{\mathsf{o}}$.}
\begin{tabular}{l|c|c|c|c|c|c}
\hline
\diagbox[dir=NW]{\makecell{$\substack{\text{Noise}\\\text{level}}$}}{\makecell{SNR}} & 5 dB     & 10 dB   & 15 dB   & 20 dB   & 25 dB   & 30 dB   \\ \hline
\hline
$\|\bm y-\bm y_{\mathsf{o}}\|_2^2$  & 34.6609 & 9.5529 & 2.8528 & 0.9720 & 0.3044 & 0.0895 \\ \hline
$\|\hat{\bm y}-\bm y_{\mathsf{o}}\|_2^2$  & 0.9798  & 0.2720 & 0.0796 & 0.0267 & 0.0087 & 0.0023 \\ \hline
From~\eqref{eq.appro_denoise_hosvd}  & 0.9286  & 0.2580 & 0.0774 & 0.0263 & 0.0082 & 0.0024 \\ \hline
\end{tabular}
\label{tab.denoise}
\end{table}

\begin{table}[t]
\centering
\scriptsize
\caption{Runtime of different schemes in seconds}
\begin{tabular}{l|c|c|c|c|c|c}
\hline
SNR & 5 dB     & 10 dB   & 15 dB   & 20 dB   & 25 dB   & 30 dB   \\ \hline
\hline
OMP  & 0.599 & 0.602 & 0.605 & 0.603 & 0.604 & 0.603 \\ \hline
cSBL  & 8.961  & 7.470 & 6.111 & 5.552 & 5.397 & 5.318 \\ \hline
AM-KroSBL  & 8.528  & 8.516 & 7.249 & 5.424 & 4.520 & 4.093 \\ \hline
SVD-KroSBL  & 4.534  & 3.360 & 2.840 & 2.668 & 2.627 & 2.608 \\ \hline
dSBL  & 0.009  & 0.005 & 0.004 & 0.004 & 0.004 & 0.004 \\ \hline
\end{tabular}
\label{tab.time_decom}
\end{table}
%\diagbox[dir=NW]{\makecell{Schemes}}{\makecell{SNR (dB)}}

\subsection{Off-grid Parameters Estimation}\label{sec.simu_offsbl}

\begin{figure}[t]
\centering
\centering
  \subcaptionbox{$M=50$\label{fig3.a}}{\includegraphics[width=0.49\linewidth]{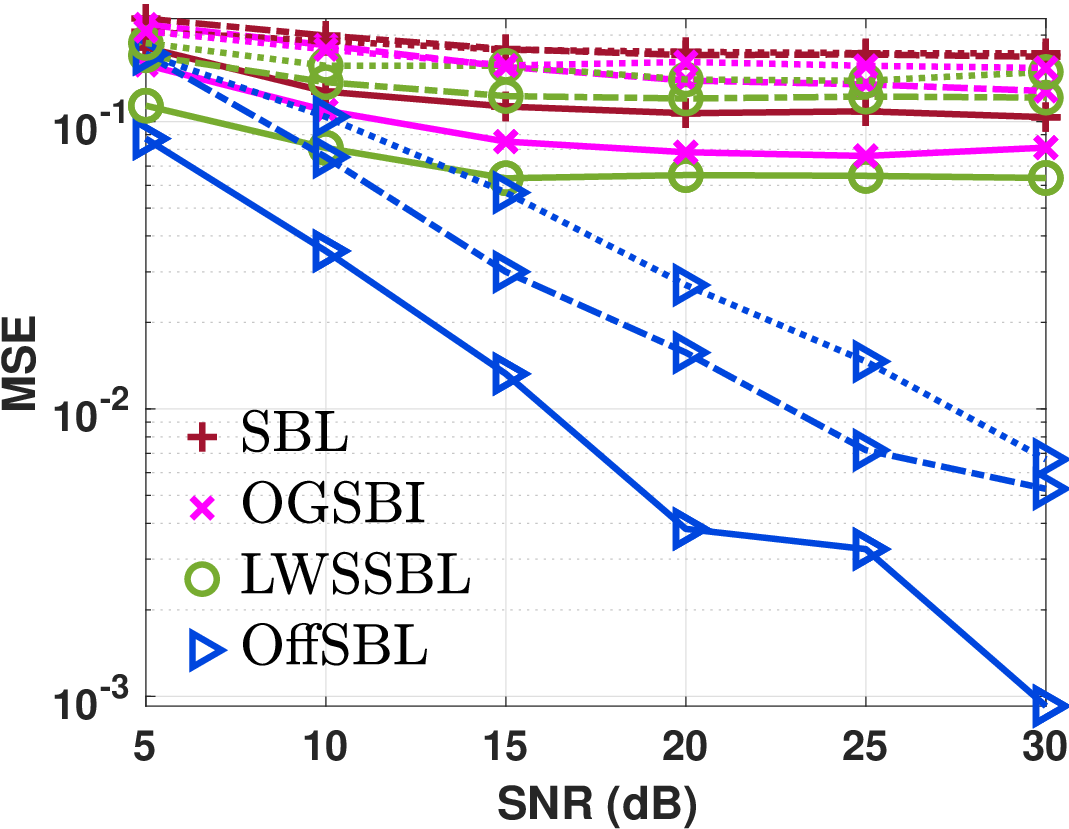}}
  \subcaptionbox{SNR $=30$dB\label{fig3.b}}{\includegraphics[width=0.49\linewidth]{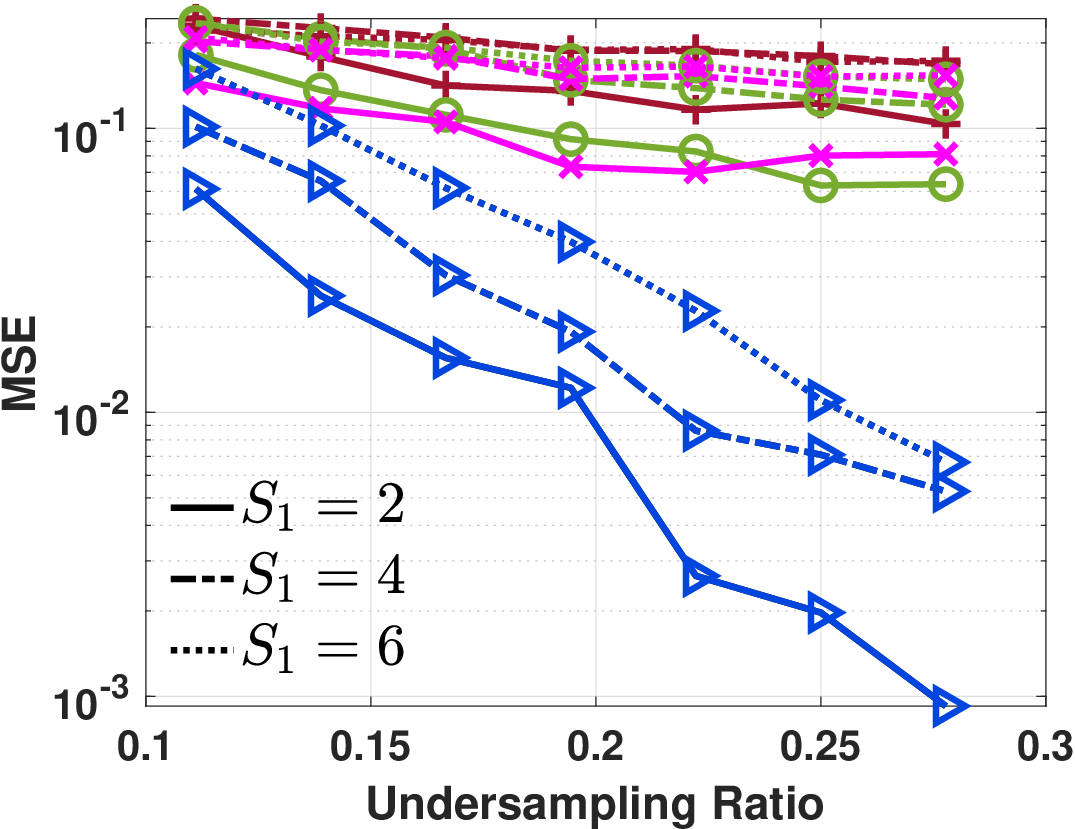}}
  \subcaptionbox{$M=50$\label{fig3.c}}{\includegraphics[width=0.49\linewidth]{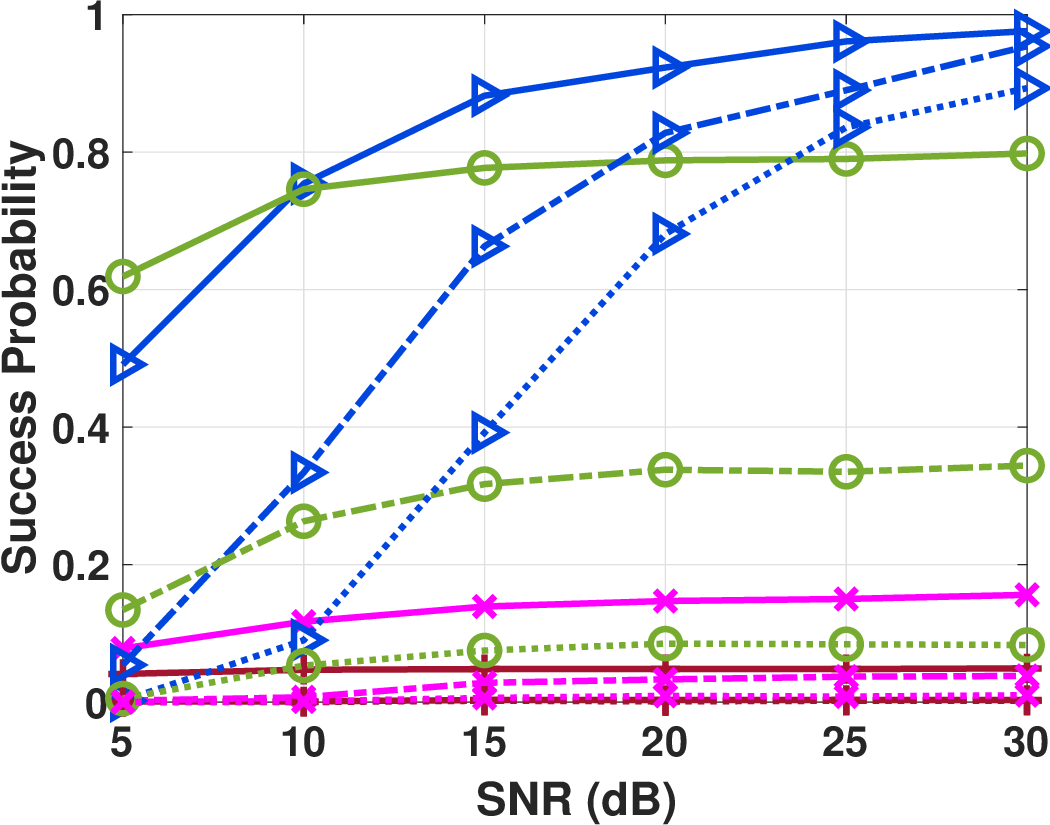}}
  \subcaptionbox{SNR $=30$dB\label{fig3.d}}{\includegraphics[width=0.49\linewidth]{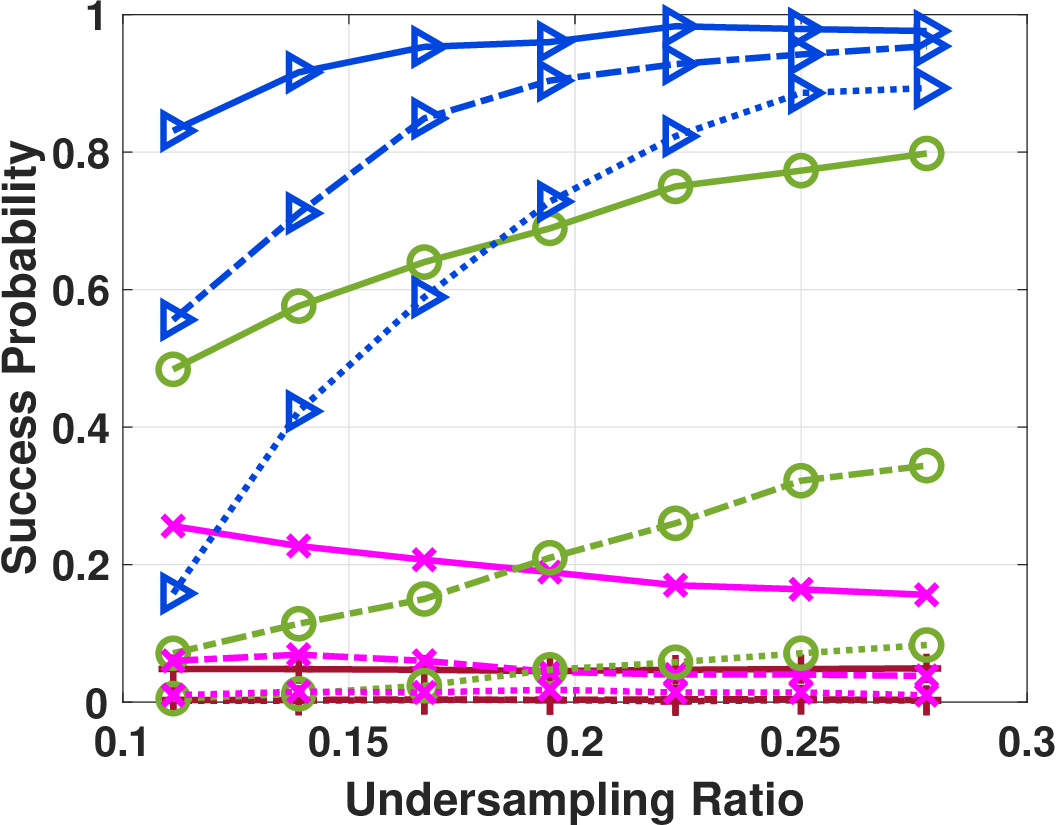}}
  
  \caption{MSE and success probability as a function of SNR and undersampling ratio $\frac{M}{N}$ for unknown parameter estimation with a varying number of unknown parameters $S_1$ and $N=180$.
  % [\color{red}In the figure, change probability to success probability; ratio to undersampling ratio; change the color of one of the blue curves to green or some other color]
  }
  \label{fig.unknown_est}
\end{figure}

\begin{table}[t]
\centering
\scriptsize
\caption{Runtime in second for unknown parameter estimation with $S_1=2$ and SNR$=30$dB.}
\begin{tabular}{l|c|c|c|c|c|c|c}
\hline
$M$ & 20      & 25    & 30    & 35    & 40    & 45    & 50  \\ \hline
OffSBL  & 0.602 & 0.546 & 0.539 & 0.541 & 0.477 & 0.504 & 0.522 \\ \hline
SBL  & 0.131  & 0.137 & 0.161 & 0.170 & 0.171 & 0.181 & 0.191 \\ \hline
LWSSBL  & 0.032  & 0.033 & 0.035 & 0.037 & 0.038 & 0.041 & 0.043 \\ \hline
OGSBI  & 0.250  & 0.269 & 0.274 & 0.286 & 0.290 & 0.305 & 0.315 \\ \hline
\end{tabular}
\label{tab.unknown_time}
\end{table}

\begin{figure*}
\centering
  \subcaptionbox{OffSBL, $\text{MSE}=0.0365$\label{fig.spec.off}}{\includegraphics[width=0.4\linewidth]{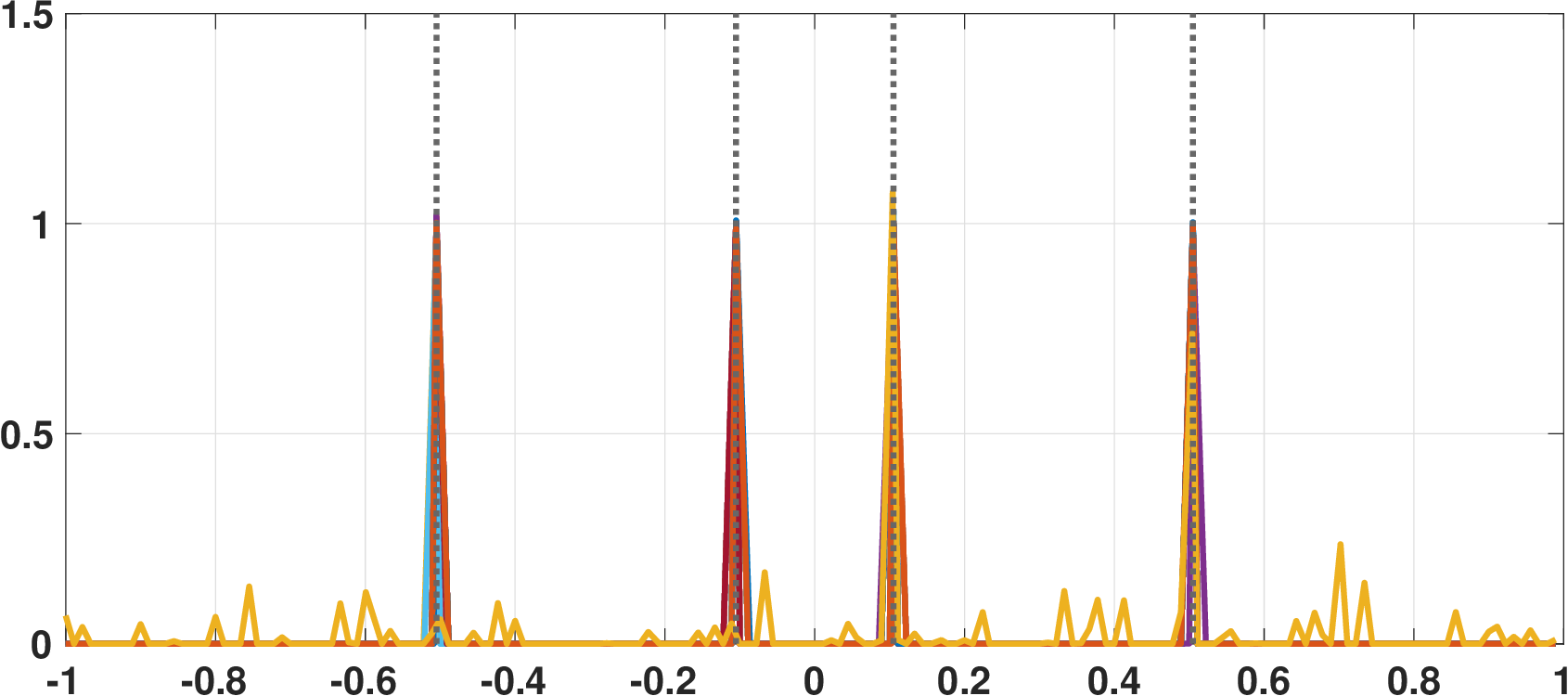}}\hspace{0em}
  \subcaptionbox{OGSBI, $\text{MSE}=0.1789$\label{fig.spec.ogsbi}}{\includegraphics[width=0.4\linewidth]{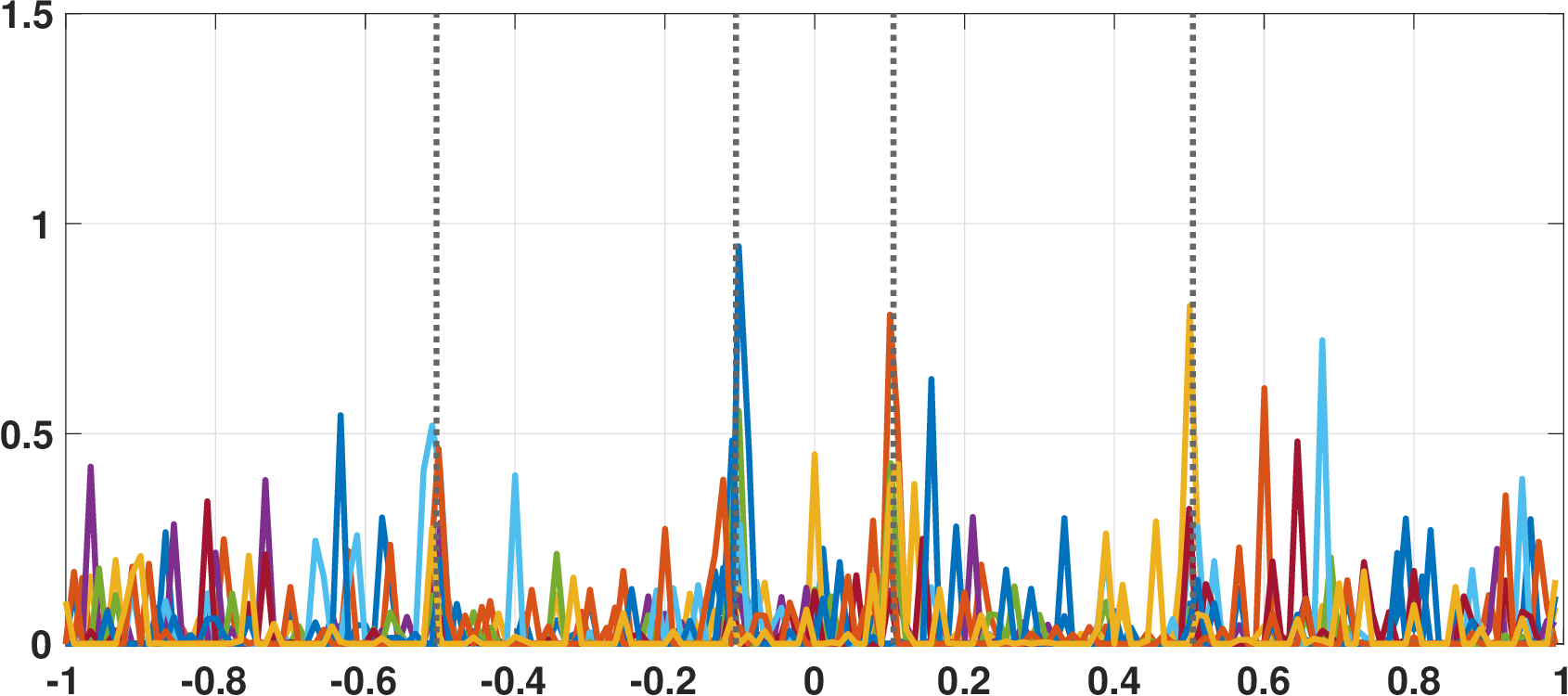}}\hspace{0em}\\
  \subcaptionbox{LWSSBL, $\text{MSE}=0.3618$\label{fig.spec.lwssbl}}{\includegraphics[width=0.4\linewidth]{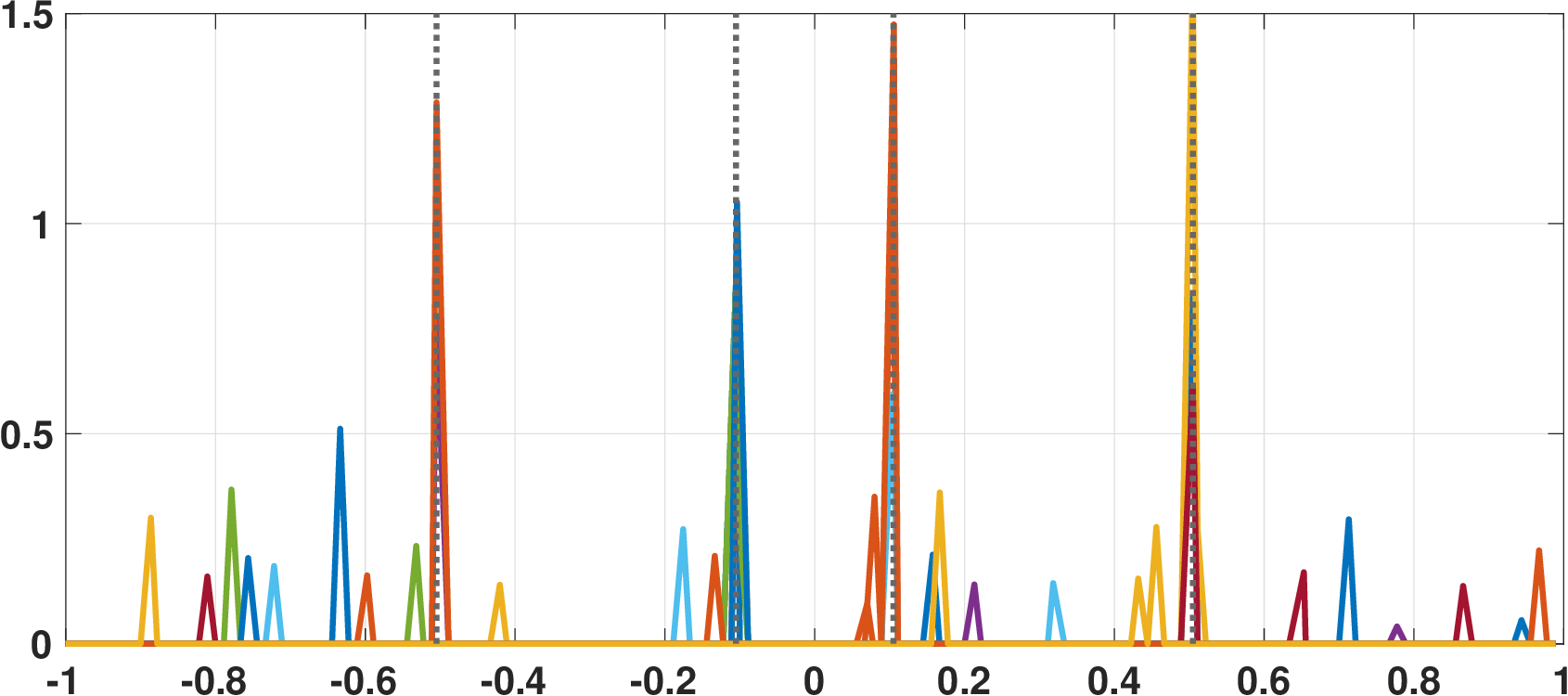}}\hspace{0em}
  \subcaptionbox{SBL, $\text{MSE}=0.2228$\label{fig.spec.on}}{\includegraphics[width=0.4\linewidth]{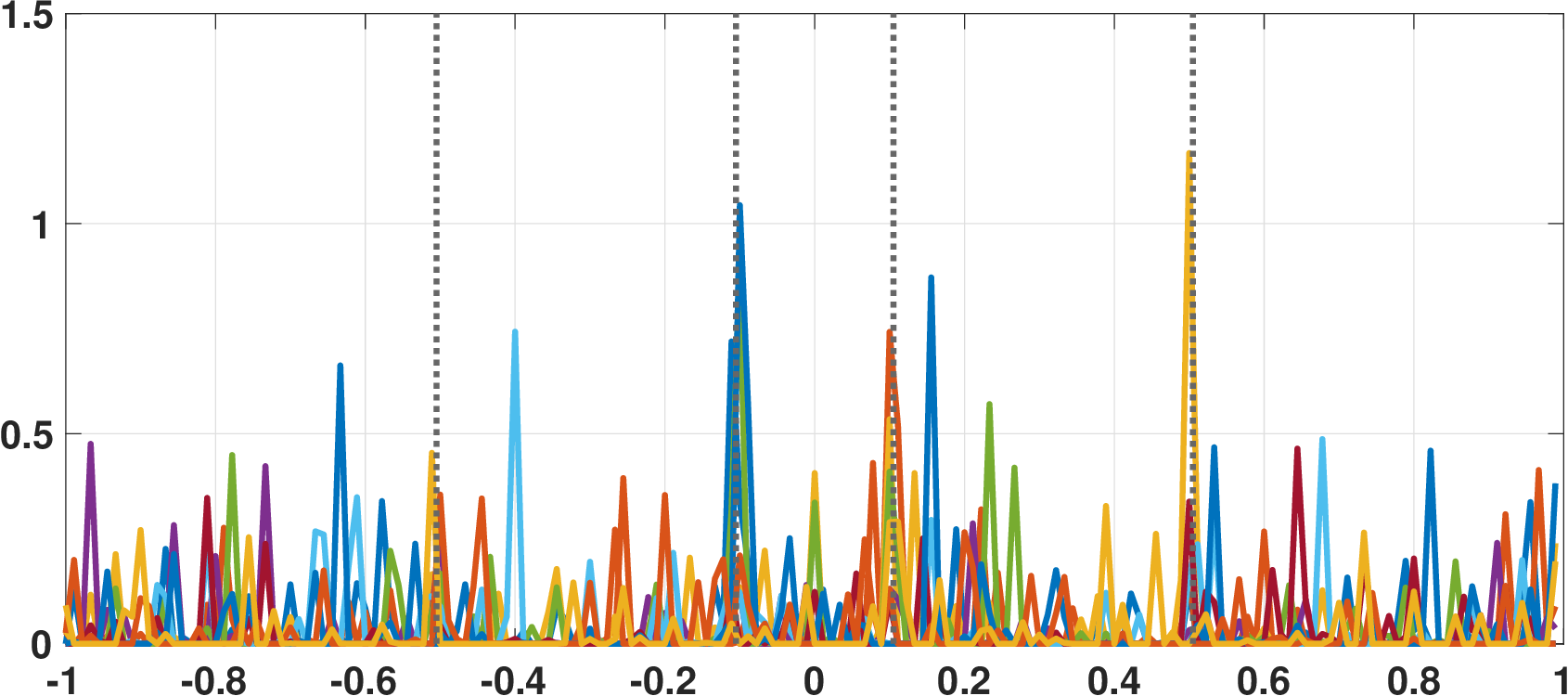}}\hspace{0em}
  \subcaptionbox{OffSBL\label{fig.wf.off}}{\includegraphics[width=0.32\linewidth]{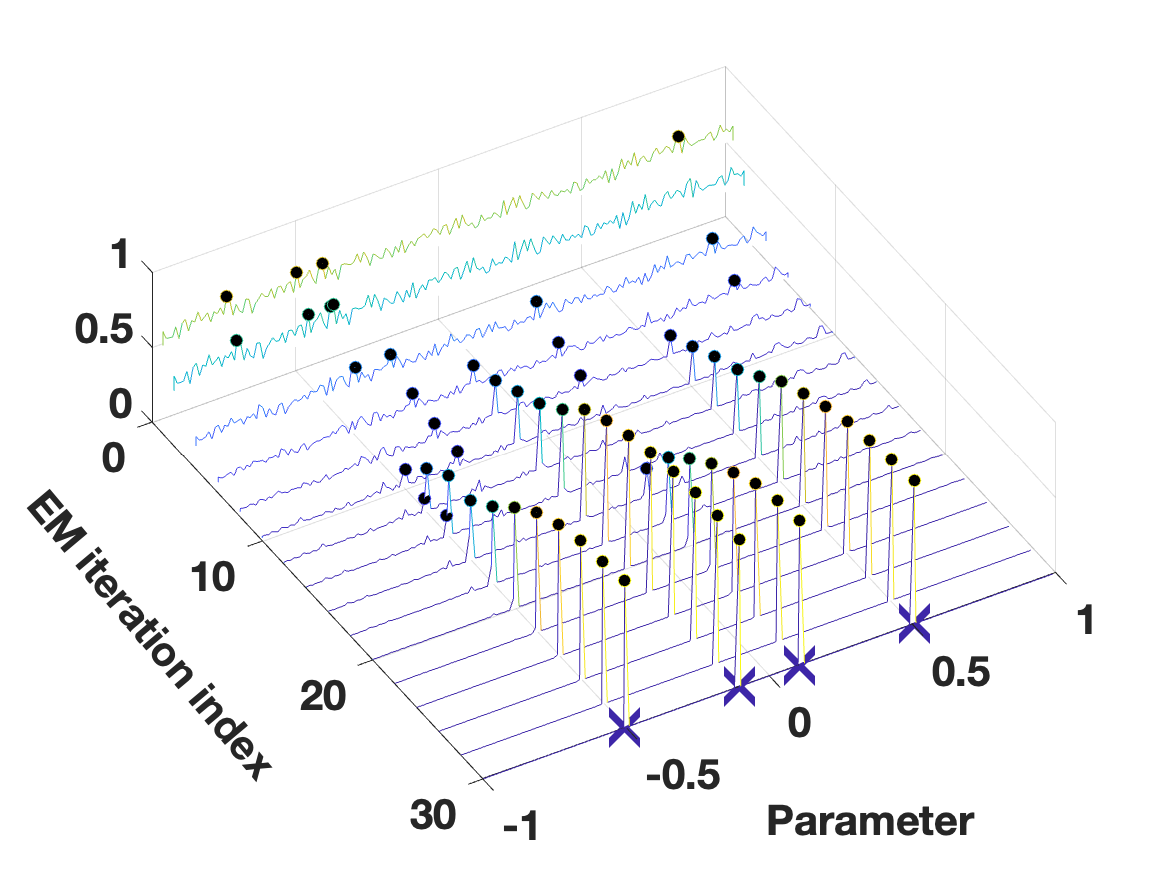}}\hspace{0em}
  \subcaptionbox{OGSBI\label{fig.wf.ogsbi}}{\includegraphics[width=0.32\linewidth]{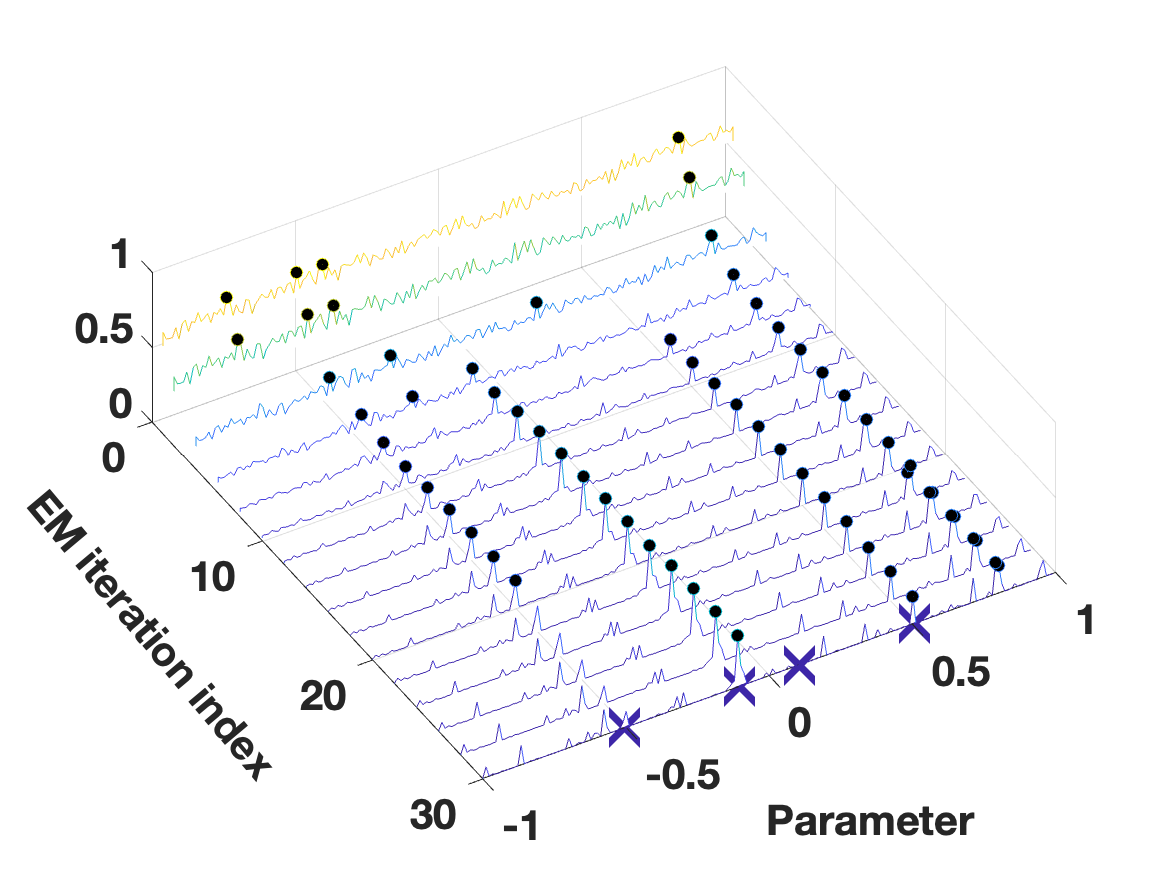}}\hspace{0em}
  \subcaptionbox{SBL\label{fig.wf.on}}{\includegraphics[width=0.32\linewidth]{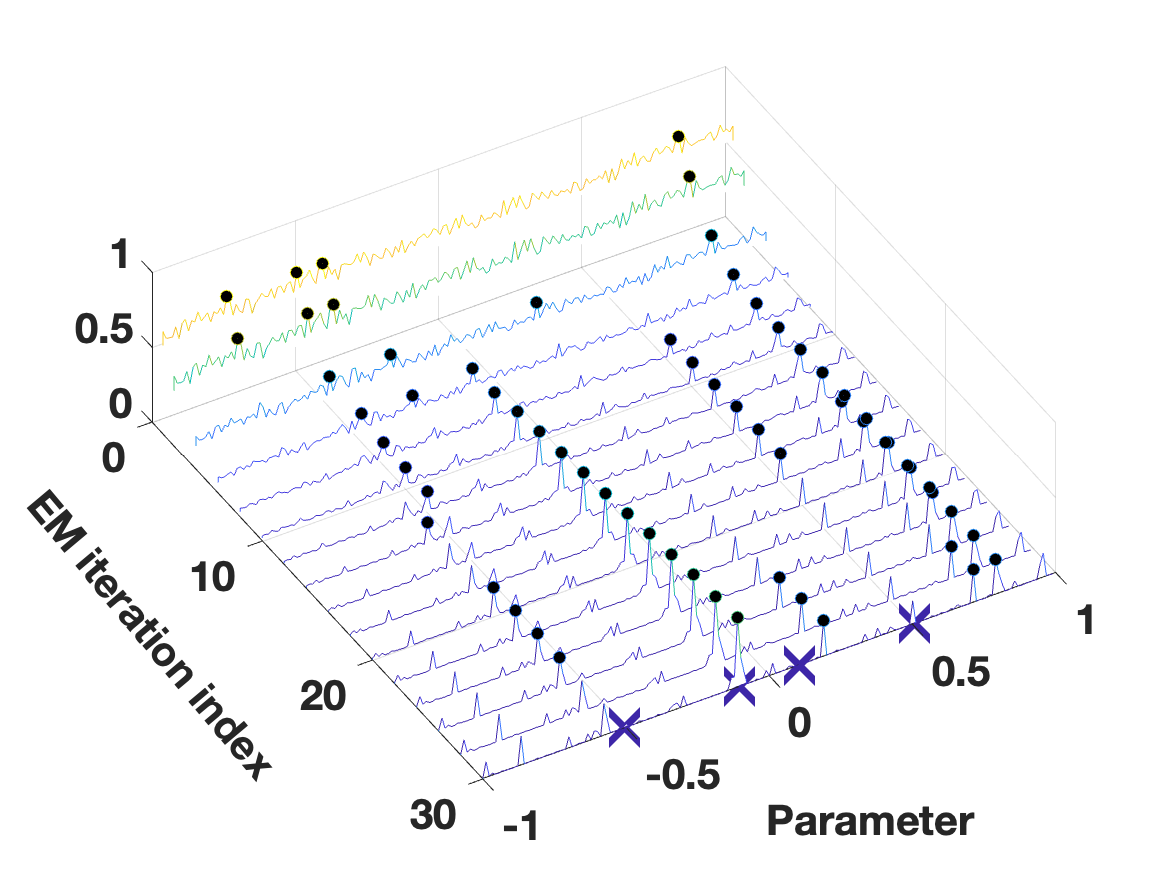}}\hspace{0em}
  \caption{Pseudospectrum $\bm \gamma$ of different algorithms in the worst-case scenario analysis. (a)-(d): pseudospectrum $\bm \gamma^{(1000)}$ of the $1000$th EM iteration. Different colors mean different independent realizations of Gaussian noise. (e)-(g): dynamic evolution of pseudospectrum $\{\bm \gamma^{(r)}\}_{1\leq r \leq 30}$ for the EM iterations of EM-based SBL algorithms. Black dot: grid points chosen to be updated. Cross: true values. Although SBL operates as an on-grid method, we still include the grid points that correspond to the top four peaks, sorely for comparison purposes.}
  \label{fig.wf_gamma}
\end{figure*}

In this section, we apply OffSBL to the unknown parameters estimation problem. The model we consider here is the case in~\eqref{eq.simplified_bem} with $i=1$, where the goal is to estimate angles $\bm \psi_{(i)}$ and coefficients $\bm x_i$. The column function is $\bm h_1(\psi)=\bm \Omega^\mathsf{T}\bm a_L(\psi)$ with $\bm \Omega \in \mathbb{C}^{L \times M}$, $L=256$, and $M$ being the number of measurements. Here, $M$ is $\{20,25,30,35,40,45,50\}$ and controls the undersampling ratio defined as $\frac{M}{N}$ with $N=180$. The matrix $\bm \Omega$ is randomly generated, whose entries take the form $e^{j\phi}$ where $\phi$ is drawn from a uniform distribution on $[0,\pi]$. We set the number of unknown parameters (angles) $S_1$ to be $\{2,4,6\}$, and the angles are drawn sequentially from a uniform distribution on $[-0.9,0.9]$ ensuring a minimal separation of $0.1$. The coefficients are drawn from $\mathcal{CN}(0,2)$. 

We use three benchmarks: (i) classical (on-grid) SBL, (ii) off-grid sparse Bayesian inference (OGSBI) using the first-order Taylor expansion~\cite{yang2012off}, and (iii) light-weight sequential SBL (LWSSBL), a state-of-the-art off-grid method using marginal likelihood optimization~\cite{pote2023light}. In our simulations, we do not provide the number of unknowns $S_1$ to all algorithms, but only an upper bound $\bar{S}$ of the number of unknowns. In practice, we only solve the problem~\eqref{eq.single_op} for the grid points corresponding to $\bar{S}$ largest peaks of the hyperparameter $\bm \gamma$ instead of all grid points. Our OffSBL algorithm estimates the noise variance using~\eqref{eq.noise_est}. SBL and OGSBI can also estimate the noise variance, while noise variance estimation for LWSSBL is not discussed \cite{pote2023light}. So for LWSSBL, we set noise variance estimate as $0.1\|\bar{\bm y}\|_2^2/M$ as in \cite{pote2023light}. We choose $\text{SNR}=10\log_{10}\mathbb{E}\{\|\bar{\bm H}_{\bar{\bm \psi}} \bar{\bm x}\|_2^2/\|\bar{\bm n}\|_2^2\}$ as $\{5,10,15,20,25,30\}$ in dB. To evaluate the performance of all schemes, we compare the mean squared error (MSE) and the success probability, where
\begin{equation*}
%    \begin{aligned}
        \text{MSE} = \mathbb{E}\left\{\frac{1}{S_1}\sum_{s=1}^{S_1}(\bar{\psi}_s - \hat{\psi}_s)^2\right\},%\\ \text{Probability} &= \frac{\# %\text{successful recovery}}{T},\notag
%    \end{aligned}
\end{equation*}
with expectation taken over $10^3$ independent trials. Here, $\bar{\psi}_s$ and $\hat{\psi}_s$ denote the true value and the estimation, respectively. The success probability is defined as the fraction of trials with MSE
smaller than~$10^{-6}$. % The results are summarized in Fig.~\ref{fig.unknown_est}.

We compare MSE and recovery probability for different SNRs and undersampling ratios in Fig.~\ref{fig.unknown_est}. We see that higher SNR and more measurements facilitate all algorithms, except OGSBI in Figs.~\ref{fig3.b} and~\ref{fig3.d}. This is because OGSBI cannot effectively optimize the grid points in this setting, as we show later in Fig.~\ref{fig.wf_gamma}. Among all candidates, our OffSBL has the best performance in both MSE and recovery probability in most cases. An exception is $\text{SNR} = 5\text{dB}$ where LWSSBL has a higher success probability. However, LWSSBL often produces larger errors when it fails, making OffSBL superior in MSE.

In Fig.~\ref{fig.wf_gamma}, we present a worst-case scenario study. We set $M = 60$ measurements and $\text{SNR} = 30\text{dB}$. The unknowns are $[-0.5050,-0.1050,0.1050,0.5050]$, shown as vertical dashed lines in Figs.~\ref{fig.spec.off}-\ref{fig.spec.lwssbl}. These values are intentionally selected to be midway between two grids to create a challenging case for grid optimization. All the coefficients are set to one. We provide the number of unknowns to all the algorithms but not the noise variance. We perform $10^3$ EM iterations to facilitate the convergence of all algorithms. The input of all algorithms is the same noiseless signal but with ten independent Gaussian noise realizations. We plot the final pseudospectrum (hyperparameter $\bm \gamma^{(1000)}$) after  $10^3$ EM iterations for ten noise realizations in Figs.~\ref{fig.spec.off}-\ref{fig.spec.lwssbl} with different colors. 

Comparing the different algorithms, we note that OffSBL consistently recovers all parameters, with minimal amplitude spikes appearing in the pseudospectrum corresponding to parameters other than the true values. LWSSBL also recovers the unknowns but with a lower success probability. LWSSBL exhibits more peaks at parameters other than the true values, implying that it is more prone to being misled by incorrect columns in the dictionary due to its greedy nature. In contrast, while OffSBL takes longer to reach the final result (see Table~\ref{tab.unknown_time}), evaluating all columns rather than proceeding greedily reduces the risk of being misled by incorrect columns. 

Further, there is little difference between OGSBI and the on-grid benchmark SBL, indicating that the first-order approximation is less effective in this case. However, OGSBI has some improvement over SBL as reflected by a lower MSE. These findings also highlight that algorithms relying on on-grid SBL for rough estimates and then refining peaks are likely to fail, as on-grid SBL often doesn’t provide a reliable starting point, with peaks rarely matching the true parameters. This is likely due to the dictionary’s structure, which takes the form $\bm \Omega^{\mathsf{T}}\bm A_{L}(\bm\psi)$ for some integer $L$. When $\bm \Omega^{\mathsf{T}}\in\mathbb{C}^{M\times L}$ has fewer rows than columns $(M<L)$, the compression effect from multiplication by $\bm \Omega^{\mathsf{T}}$ can lead to information loss, creating a challenging setting for off-grid sparse recovery~\cite{guo2018doa}. However, in many applications, such as IRS channel estimation, where the value of $M$ represents the number of time slots, $M$ is typically limited. Thus, integrating grid updates into the EM iteration, as implemented in OffSBL, is essential.  

We further present the pseudospectrum $\{\bm \gamma^{(r)}\}_{1\leq r \leq 30}$ for OffSBL, SBL, and OGSBI, along with the grid points that are updated dynamically throughout the EM iteration in Figs.~\ref{fig.wf.off}-\ref{fig.wf.ogsbi}. Although SBL is an on-grid method, we still pinpoint the grids of the top four peaks. All algorithms have the same $\bm \gamma$ initially which becomes different afterward. Our OffSBL demonstrates superior optimization of grid points, identifying the correct values and amplitudes, whereas SBL and OGSBI do not reveal the true parameters. The evolution of the pseudospectrum across EM iterations highlights the effectiveness of our grid adjustment step.

\subsection{IRS-aided Wireless Channel Estimation}
\begin{figure}
\centering
  \subcaptionbox{Angle estimation performance\label{fig.mse}}{\includegraphics[width=0.49\linewidth]{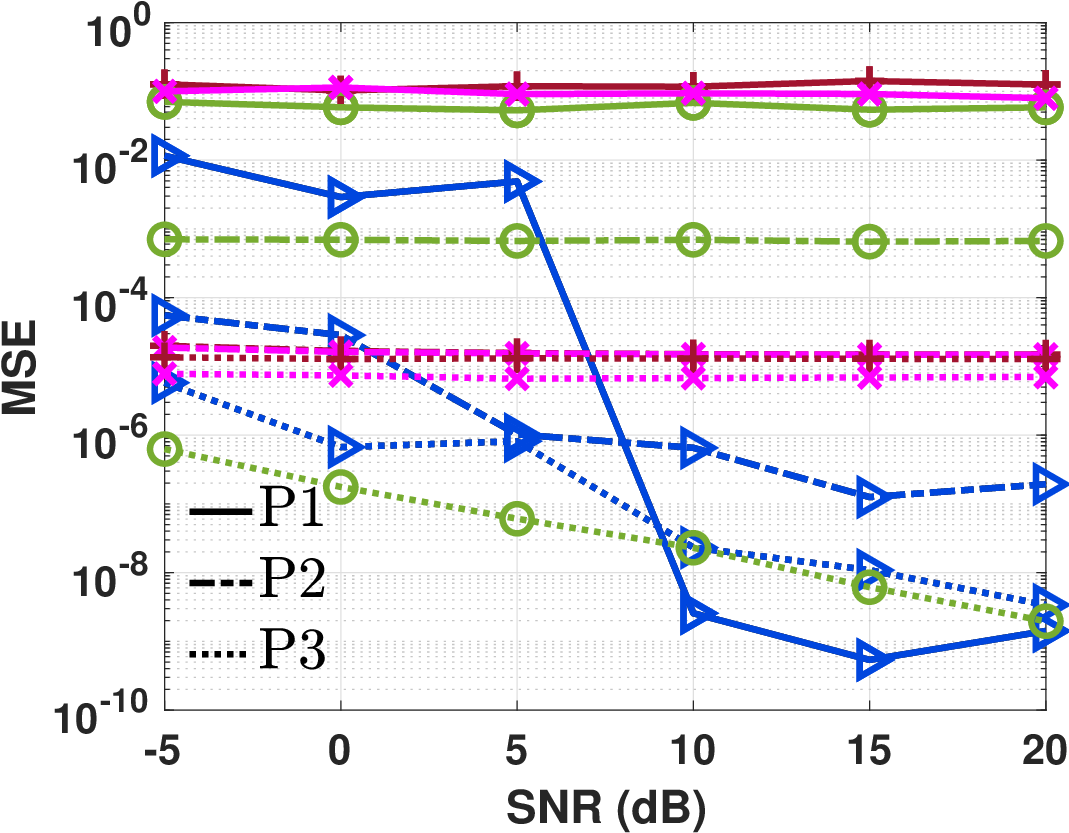}\includegraphics[width=0.49\linewidth]{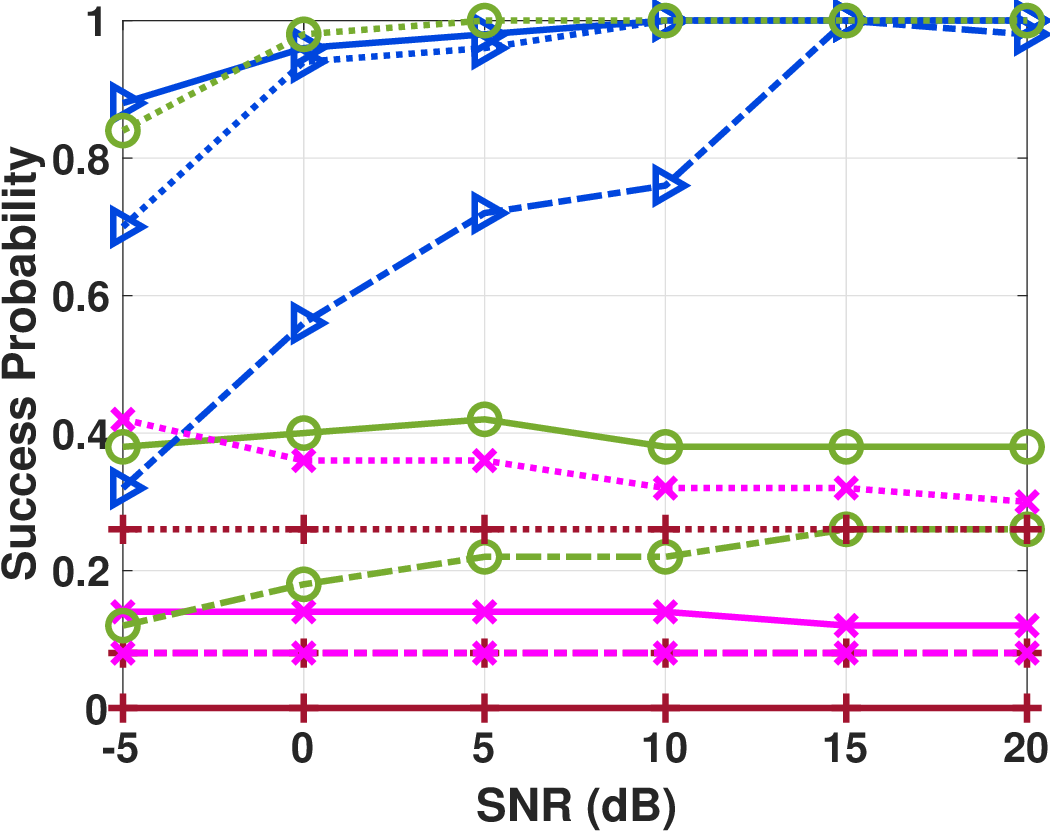}}
  % \subcaptionbox{Probability\label{fig.prob}}{\includegraphics[width=0.48\linewidth]{figure/ce_p123_prob.eps}}\hspace{0em}
  % \subcaptionbox{NMSE\label{fig.ce.a}}{\includegraphics[width=0.49\linewidth]{figure/ce_nmse.eps}}\hspace{0em}%
  \subcaptionbox{Channel estimation performance\label{fig.ce.b}}{\includegraphics[width=0.49\linewidth]{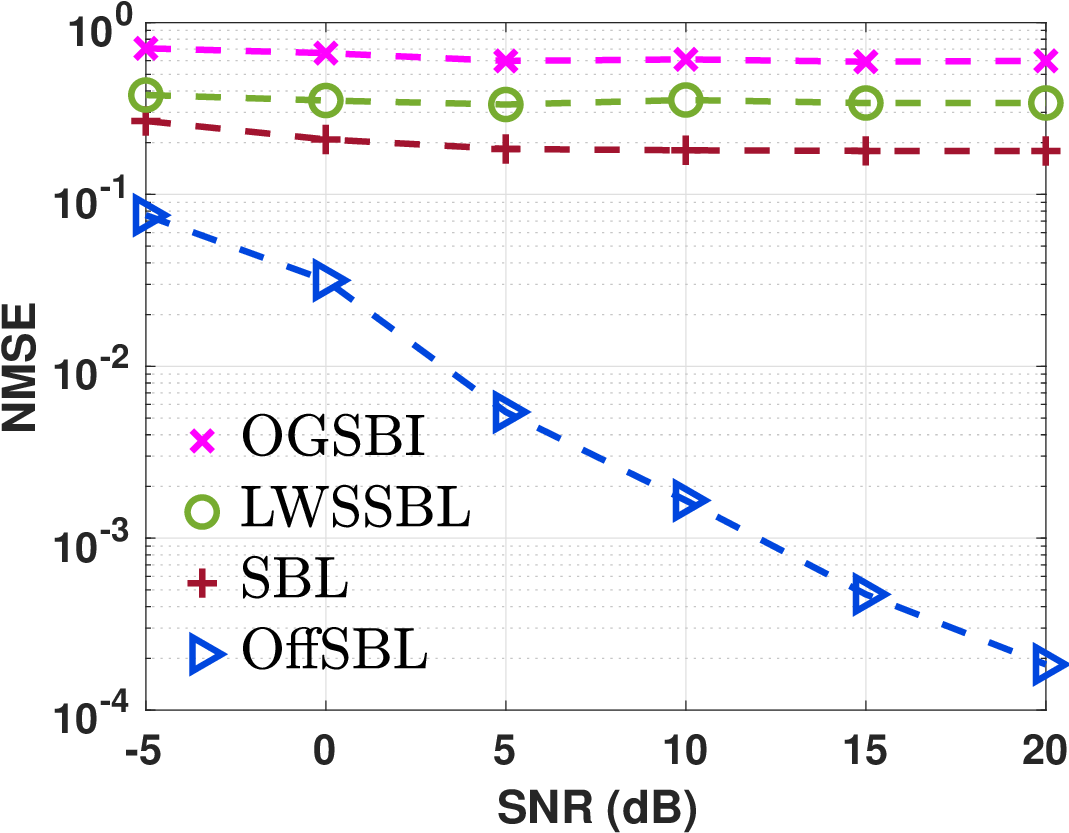}\includegraphics[width=0.49\linewidth]{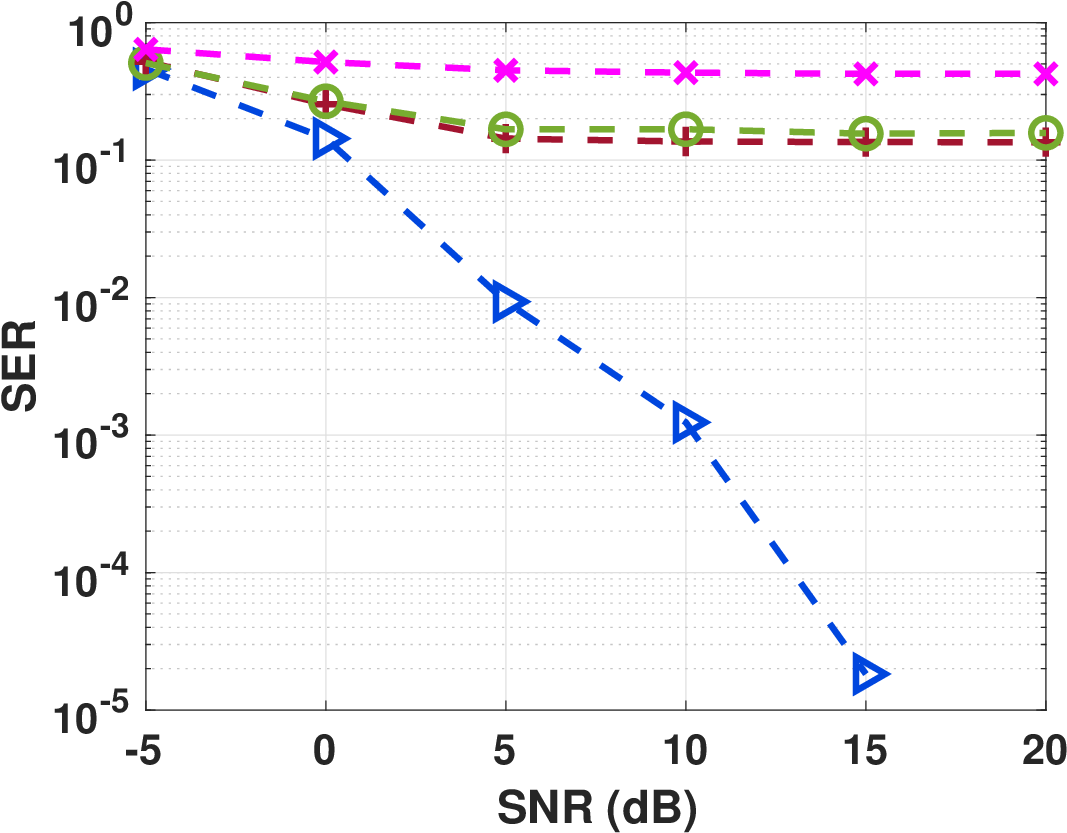}}
  % \hspace{0em}  
  \caption{IRS-aided channel estimation. (a): MSE and recovery success probability for angle estimation in P1, P2, and P3. (c): NMSE of IRS-aided channel estimation and SER of different algorithms as functions of SNR.}
  \label{fig.ce_recovery}
\end{figure}
We focus on the IRS-aided channel estimation problem, as described in Sec.~\ref{sec.channelmodel}. Here, we first use the decomposition step and then turn to the BEM and apply OffSBL separately for $i=1,2,3$ in~\eqref{eq.simplified_bem}. Thus, the channel estimation scheme can be viewed as a collective evaluation of the decomposition and OffSBL. For benchmarking, we apply the same decomposition step, and then solve~\eqref{eq.simplified_bem} using the same algorithms as in Sec. \ref{sec.simu_offsbl}. For simplicity, we denote the problem~\eqref{eq.simplified_bem} with $i=1,2,3$ as P1, P2, and P3, respectively.

For IRS-aided channel estimation, we use $R=16$ BS antennas, $T=6$ MS antennas, $L=256$ IRS elements. We consider only one path between the BS and IRS\cite{dampahalage2022supervised,he2020channel,wan2021terahertz,yashvanth2022cascaded}, as the IRS is typically mounted in locations with fewer obstacles\cite{zheng2022survey,liu2021reconfigurable}, and the line-of-sight path is generally much stronger than the other paths. Therefore, we take $P_{\mathrm{BS}} = 1$ and $P_{\mathrm{MS}} = 3$. The IRS configuration entries $\{\bm \omega_k\}_{k=1}^{K_{\mathrm{I}}}$ are $1/\sqrt{L}e^{j\phi}$ where $\phi$ is drawn uniformly randomly from $[0,\pi]$. with $K_{\mathrm{I}}=40$. We send $K_{\mathrm{P}}=20$ pilot signals for each IRS configuration.
%We note that compared to our previous work \cite{he2022structure}, here we send more pilot signals. This is because we deal with an off-grid problem while \cite{he2022structure} is an on-grid problem. 
For our OffSBL algorithm, the dictionaries in P1, P2, and P3 are constructed by $N_1 = 180$, $N_2 = 50$, and $N_3 = 50$, respectively. For the other algorithms, dictionaries are constructed using $N_1 = 180$, $N_2 = 150$, and $N_3 = 150$ grid points. The channel gains $\beta_{{\mathrm{BS}}}$ and $\{\beta_{{\mathrm{MS}},p}\}_{p=1}^{P_{\mathrm{MS}}}$ in \eqref{eq.channelmodel1} and \eqref{eq.channelmodel2} are drawn from the standard complex Gaussian distribution~\cite{lin2021channel}. We randomly draw $\alpha_{\mathrm{MS}}$, $\{\phi_{\mathrm{MS},p}\}_{p=1}^{P_{\mathrm{MS}}}$, $\phi_{\mathrm{BS}}$, and $\alpha_{\mathrm{BS}}$ from uniform distribution in $[0.3,0.5]$, $[-0.2,0.2]$, $[0.3,0.5]$, and $[0,0.5]$, respectively. We also assume that the angles, after being spread, are separated by at least $0.07$. We opt for SNR $\{-5,0,5,10,15,20\}$ in dB. Along with MSE and success probability of the angle (parameter) estimation, we also use NMSE and symbol error rate (SER) to evaluate the channel estimation performance, where NMSE is 
\begin{equation*}
    \frac{1}{K_{\mathrm{I}}}\sum_{k = 1}^{K_{\mathrm{I}}}\!\!\frac{\|\bm H_\mathrm{BS} \diag (\bm \omega_{k}) \bm H_\mathrm{MS} - \tilde{\bm H}_\mathrm{BS} \diag (\bm \omega_{k}) \tilde{\bm H}_\mathrm{MS}\|_\mathrm{F}^2}{\|\bm H_\mathrm{BS} \diag (\bm \omega_{k}) \bm H_\mathrm{MS}\|_\mathrm{F}^2}
\end{equation*}
with $ \tilde{\bm H}_\mathrm{BS} \diag (\bm \omega_{k}) \tilde{\bm H}_\mathrm{MS}$ being the channel estimate. We compute SER using $10^6$ $16$-QAM symbols decoded using the estimated channel.
% The results in Fig.~\ref{fig.ce_recovery}.% are averaged in fifty trials.

We first examine the angle estimation results in Fig.~\ref{fig.ce_recovery}. It can be seen that our OffSBL can achieve the best performance in solving P1 and P2 except for the low SNR case for P2. P3 reduces to the normal DoA estimation problem, where LWSSBL exhibits superior recovery ability. However, at higher SNRs, our algorithm is able to achieve comparable performance. As evident from the NMSE and SER plots, OffSBL consistently recovers the true angles and accurately retrieves the coefficients, leading to the best NMSE and SER. Although other algorithms can perform well in solving P2 and P3, the significant recovery errors in P1 affect the overall accuracy of channel retrieval.

\section{Conclusion}

We addressed the joint estimation of unknown parameters and coefficients from Kronecker-structured measurements, focusing on IRS-aided wireless channel estimation. Leveraging the Kronecker structure, we decomposed the problem into smaller independent subproblems. Each subproblem was solved with EM-based SBL integrated with a novel grid optimization method to reduce grid mismatch. We provided a theoretical analysis of the error bound for the decomposition step and established the algorithm's convergence. Our decomposition step also reduces the noise level in the measurements, which was also analyzed theoretically. Numerical results showed that the decomposition step reduces complexity, while the grid optimization improves accuracy. Future work 
can consider analyzing the resolution of the OffSBL method and extending our framework to the recovery of sparse tensors with ranks greater than one.
%could refine the OffSBL method and strengthen the theoretical guarantees of the algorithm.

% \section*{Acknowledgments}
% This should be a simple paragraph before the References to thank those individuals and institutions who have supported your work on this article.

\appendices

\section{Proof of Theorem~\ref{thm.angle_denoise}}\label{appe.angle_and_denoising}
We define $\bm P_{\bm U}$ as the projection matrix onto the column space of a given matrix $\bm U$ and $\bm U_\perp$ as the projection onto its orthogonal subspace. Also, $\|\cdot\|$ is the matrix spectral norm. We need the below lemma for the proof.
\begin{lemma}\label{lm.svd_deco}
    \cite[Supplement Sec. 1.2]{cai2018rate} Suppose $\bm X \in \mathbb{R}^{p_1 \times p_2}$ is a rank-$r$ matrix and $\bm Y = \bm X + \bm Z$ where the entries of $\bm Z \in \mathbb{R}^{p_1 \times p_2}$ follow a zero mean Gaussian distribution with unit variance. We denote $\bm V,\hat{\bm V} \in \mathbb{R}^{p_2 \times r}$ as the matrix of the right singular vectors of $\bm X$ and the matrix of the top $r$ right singular vectors of $\bm Y$, respectively. Suppose the $r$th right singular value $\sigma_r^2$ of  $\bm X$ satisfies $\sigma_r^2 \geq C_{\mathrm{gap}} \sigma^2( \sqrt{p_1 p_2}+ p_2 )$ for some large constant $C_{\mathrm{gap}}>0$. Then, for all $x\geq 0$, there exist constants $C,c>0$ such that
    \begin{multline*}
        \mathbb{P} \left\{ \| \bm P_{\bm Y \bm V} \bm Y \bm V_\perp \| \leq x \right\} \geq 1- C \exp \left\{ -c (\sigma_r^2 + p_1) \right\}
        \\- C \exp \left\{ C p_2 - c \min \left( x^2, x \sqrt{\sigma_r^2 + p_1} \right) \right\},
    \end{multline*}
    and with probability exceeding $1-C \exp \left\{ -c \sigma_r^4/(\sigma_r^2 + p_1) \right\}$, 
    \begin{equation*}
        \| \sin \Theta(\hat{\bm V}, \bm V) \|^2 \leq C (\sigma_r^2 + p_1) \sigma_r^{-4} \| \bm P_{\bm Y \bm V} \bm Y \bm V_\perp \|^2,
    \end{equation*}
    Here, $\Theta(\hat{\bm V}, \bm V)=\diag(\arccos(\sigma_1),\cdots,\arccos(\sigma_r))$ where $\sigma_1\geq\cdots\geq\sigma_r\geq 0$ are the singular values of $\bm V^\mathsf{T}\hat{\bm V}$. 
\end{lemma}

We prove Theorem~\ref{thm.angle_denoise} for $i=1$ using Lemma~\ref{lm.svd_deco}, and $i=2,\cdots,I$ follow similarly. Also, we consider the decomposition of $\bar{\bm y}/\sigma_\mathsf{t}$ instead of $\bar{\bm y}$. This scaling does not alter the subspaces obtained after decomposition but ensures that the noise entries follow a zero-mean, unit-variance Gaussian distribution, consistent with Lemma~\ref{lm.svd_deco}. 

For $i=1$, the true and estimated subspaces are spanned by $\bm y_1/\Vert\bm y_1\Vert$ and $\hat{\bm y}_1/\Vert\hat{\bm y}_1\Vert$, respectively. The first mode matricization of the tensor $\mathcal{\bm Y}$, as defined in~\eqref{eq.matricization}, is $\bm Y_{(1)}$. Setting $r=1$, $p_1=\bar{M}/M_1$ and $p_2=M_1$, and consequently, $\sigma_r = \lambda$ in Lemma~\ref{lm.svd_deco}, we derive
\begin{multline}
        \sin^2\vartheta_1= \|\sin\Theta (\hat{\bm y}_1/\Vert\hat{\bm y}_1\Vert,\bm y_1/\Vert\bm y_1\Vert)\| 
        \\\leq C(\lambda^2+\bar{M}/M_1)\lambda^{-4}\|\bm P_{\bm Y_{(1)}^\mathsf{T}\bm y_1} \bm Y_{(1)}^\mathsf{T} \bm y_{1_\perp}\|^2,\label{eq.upp_prob_1}
\end{multline}
with probability at least $1-C e^{-c\frac{\lambda^4}{\lambda^2+p_1}}$. 

Further, we bound $\|\bm P_{\bm Y_{(1)}^\mathsf{T}\bm y_1} \bm Y_{(1)}^\mathsf{T} \bm y_{1_\perp}\|^2$ using Lemma~\ref{lm.svd_deco} by setting $x=\sqrt{\tilde{C} M_1}$ where $\tilde{C}<C/c<C_{\mathrm{gap}}$,
\begin{multline}\label{eq.boundP}
        \mathbb{P} \left\{ \| \bm P_{\bm Y \bm V} \bm Y \bm V_\perp \| \leq \tilde{C}\sqrt{M_1} \right\} \geq 1- C \exp \left\{ -c (\sigma_r^2 + p_1) \right\}
        \\- C \exp \left\{ C M_1 - c \min \left( \tilde{C} M_1,  \sqrt{\tilde{C} M_1(\sigma_r^2 + p_1)} \right) \right\},
    \end{multline}
     Then, we simplify the right-hand side of \eqref{eq.boundP} using 
\begin{equation}\label{eq.boundlambda}
    \sigma_r^2+p_1\geq \sigma_r^2= \lambda^2 \geq C_{\mathrm{gap}} ( \sqrt{p_1M_1}+ M_1)\geq C_{\mathrm{gap}} M_1. 
\end{equation}
Then, \eqref{eq.boundP} is simplified as 
\begin{multline*}\label{eq.boundP}
        \mathbb{P} \left\{ \| \bm P_{\bm Y \bm V} \bm Y \bm V_\perp \| \leq \tilde{C}\sqrt{M_1} \right\} \geq 1-C \exp \left\{ -c C_{\mathrm{gap}}  M_1  \right\}
        \\-C \exp \left\{ \left[C - c \min ( \tilde{C}, \sqrt{\tilde{C}C_{\mathrm{gap}} })\right]M_1 \right\} \geq 1- C e^{-\tilde{c} M_1},
    \end{multline*}
% The second term on the right-hand side of \eqref{eq.boundP} simplifies as 
% \begin{equation}
% C\exp \left\{ -c (\sigma_r^2 + p_1) \right\}
% \leq C \exp \left\{ -c C_{\mathrm{gap}}  M_1  \right\}.\label{eq.prob_bound1}
% \end{equation}
% Similarly, the last term on the right-hand side of \eqref{eq.boundP} is
% % \begin{align}
% % C \exp \left\{ C M_1 - c \min \left( \tilde{C} M_1,  \sqrt{\tilde{C} M_1(\sigma_r^2 + p_1)} \right) \right\}\notag\\
% % &\hspace{-6.5cm}\leq C \exp \left\{ C M_1 - c \min \left( \tilde{C}M_1, \tilde{C}\sqrt{M_1} \sqrt{C_{\mathrm{gap}} M_1} \right) \right\}\notag\\
% % &\hspace{-6.5cm}\leq C \exp \left\{ \left[C - c \min ( \tilde{C}, \tilde{C}\sqrt{C_{\mathrm{gap}} })\right]M_1 \right\}.\label{eq.prob_bound2}
% % \end{align}
% \begin{multline}
% C \exp \left\{ C M_1 - c \min \left( \tilde{C} M_1,  \sqrt{\tilde{C} M_1(\sigma_r^2 + p_1)} \right) \right\}\\
% \leq C \exp \left\{ \left[C - c \min ( \tilde{C}, \tilde{C}\sqrt{C_{\mathrm{gap}} })\right]M_1 \right\}.\label{eq.prob_bound2}
% \end{multline}
for some constant $\tilde{c}$.
% Combining \eqref{eq.prob_bound1}and \eqref{eq.prob_bound2}, we deduce from \eqref{eq.boundP} that there exist constants $C$, $\tilde{C}$ and $\tilde{c}$ such that 
% \begin{equation}
%     \mathbb{P} \left\{ \|\bm P_{\bm Y_{(1)}^\mathsf{T}\bm y_1} \bm Y_{(1)} \bm y_{1_\perp}\| \leq \tilde{C}\sqrt{M_1} \right\} \geq 1- C e^{-\tilde{c} M_1}.
% \end{equation}
So, \eqref{eq.upp_prob_1} and the union bound implies 
\begin{equation}
    \sin^2\vartheta_1 \leq \frac{C(\lambda^2+\bar{M}/M_1)}{\lambda^{4}}\tilde{C}M_1 = \frac{C\tilde{C}M_1}{\lambda^2}+\frac{C\tilde{C}\bar{M}}{\lambda^4},\label{eq.boundsin}
    \end{equation}
with probability exceeding $1-C e^{-c\frac{\lambda^4}{\lambda^2+p_1}}-C e^{-\tilde{c} M_1}$. 

Furthermore, since $\lambda^2\geq C_{\mathrm{gap}} ( \sqrt{p_1M_1}+ M_1)$, with $\bar{C}=\min(C_{\mathrm{gap}},C_{\mathrm{gap}}^2)$, we derive
\begin{equation*}
    \lambda^2\geq \sqrt{\bar{C}p_1M_1}+ \bar{C}M_1\geq \frac{ \bar{C}M_1+\sqrt{\bar{C}^2M_1^2+4\bar{C}p_1M_1}}{2}.
\end{equation*}
Consequently, $\lambda^4-\bar{C}M_1\lambda^2-\bar{C}p_1M_1\geq0$ because $\lambda^2$ is greater than both roots of the quadratic function in $\lambda
^2$. Thus, we deduce 
$\frac{\lambda^4}{\lambda^2+p_1}\geq \bar{C}M_1 $. So, from \eqref{eq.boundsin}, we arrive at the desired result,
\begin{equation*}
    \mathbb{P}\left\{\!\sin\vartheta_1 \!\leq \!\frac{\sqrt{C\tilde{C}M_1}}{\lambda}+\frac{\sqrt{C\tilde{C}\bar{M}}}{\lambda^2}\!\right\}\geq 1-2C e^{-\max(c\bar{C},\tilde{c}) M_1}.
    \end{equation*}

\section{Proof of Theorem \ref{thm.denoising}}\label{appe.thm_denoising}

%We define $\bm P_{i}:=\hat{\bm{y}}_i \hat{\bm{y}}_i^\mathsf{T}$ as the projection matrix to $\hat{\bm{y}}_i$ and $\bm P_{i_\perp}$ as the orthogonal projection matrix. For a tensor $\mathcal{X}$, its $i$th mode product with matrix $\bm U$ is defined as $\mathcal{X} \times_i \bm U\coloneq\bm U \bm{X}_{(i)}$, where $\bm{X}_{(i)}$ is the $i$th mode matricization.
Our proof uses the following lemmas.
\begin{lemma}\label{lm.pertubed_x}
\cite[Lemma 6]{zhang2018tensor}: Suppose $\bm X,\bm Z$ are two matrices, and the projection matrix orthogonal to the subspace spanned by the leading $r$ left singular vectors of $\bm X + \bm Z$ is $\bm U$. Then, $\| \bm U \bm X \|_\mathrm{F} \leq 2 \sqrt{r} \|\bm Z \|$, where $r$ is the rank of $\bm X$.
\end{lemma}
\begin{lemma}\label{lm.gaussian_spectral}
    \cite[Corollary 5.35]{vershynin2010introduction} Let $\bm Z \in \mathbb{R}^{p_1\times p_2}$ whose entries are independent Gaussian random variables with zero mean unit variance. Then, for any $x \geq 0$, the matrix satisfies $\|\bm Z\| \geq \sqrt{p_1} + \sqrt{p_2} + x$, with probability less than $2e^{-x^2/2}$.
\end{lemma}
\begin{lemma}\label{lm.chi_square}
\cite[Lemma 8.1]{birge2001alternative}
Suppose $X$ satisfies the non-central $\chi^2_d(\nu)$ distribution with $d$ degrees of freedom and non-centrality parameter $\nu$.
%, i.e., $X=\sum_{i=1}^d n_i^2$, where $n_i$ follows Gaussian distribution with mean $\mu_i$ and unit variance, and $\nu=\sum_{i=1}^d \mu_i^2$. 
Then, for all $x>0$, it satisfies $X \geq (d + \nu) + 2\sqrt{(d + 2\nu)x} + 2x$ with probability less than $e^{-x}$.
%Suppose $X$ satisfies the non-central $\chi^2$ distribution with $d$ degrees of freedom and non-centrality parameter $\nu$, i.e., $X=\sum_{i=1}^d n_i^2$, where $n_i$ follows Gaussian distribution with mean $\mu_i$ and unit variance, and $\nu=\sum_{i=1}^d \mu_i^2$. Then, for all $x>0$, it is bounded as $X \geq (d + \nu) + 2\sqrt{(d + 2\nu)x} + 2x$ with probability less than $e^{-x}$.
\end{lemma}
\begin{lemma}\cite{de2000multilinear}\label{lm.mixedproduct}
Consider tensors $\mathcal{Y}, \mathcal{X}$, and $\mathcal{Z}$, such that $\mathcal{\bm Y} = \mathcal{X} \times_1 \bm A_1\cdots\times_I\bm A_I$, where $\mathcal{X} = \mathcal{Z} \times_1 \bm B_1\cdots \times_I \bm B_I$, for any compatible matrices $\{\bm A_i,\bm B_i\}_{i=1}^I$. Then, $\mathcal{\bm Y} = \mathcal{Z} \times_1 \times_1 (\bm A_1 \bm B_1)\cdots \times_I (\bm A_I \bm B_I)$.
\end{lemma}

% \subsection*{Proof of Theorem \ref{thm.denoising}}
Setting the tensor order $R=I$ and $r$-ranks as $p_i = 1$ for $i\in[I]$ in \cite[Eq. (19)]{balda2016first} leads to~\eqref{eq.appro_denoise_hosvd}. To prove the probabilistic bound, we note that using tensor notation, $\|\otimes_{i=1}^I \hat{\bm y}_i-\otimes_{i=1}^I \bm y_i\|_2=\Vert \hat{\mathcal{\bm Y}} - \mathcal{\bm Y} \Vert_\mathrm{F}$, where $\hat{\mathcal{\bm Y}}$ is HOSVD output and $\mathcal{\bm Y} = \mathcal{\bm Y} + \bar{\mathcal{N}}$, with $ \mathcal{\bm Y},\mathcal{\bm Y}$ and $\bar{\mathcal{N}}$ being the measurement, the noiseless signal, and noise tensors, respectively, from \eqref{eq.Kro_generate_data}.

HOSVD reconstructs $\hat{\mathcal{\bm Y}} = \xi \times_1 \bm e_1\cdots\times_I\bm e_I$. Here, $\bm e_i$, the leading left singular vectors of the $i$th mode matricization of $\bar{\mathcal{Y}}$ given by ${\bm Y}_{(i)} +\bm N_{(i)}$, with ${\bm Y}_{(i)}$ and $\bm N_{(i)}$ are the $i$th mode matricization of $\mathcal{Y}$ and $\bar{\mathcal{N}}$, respectively. Also, $\xi = \bar{\mathcal{\bm Y}} \times_1 \bm e_1^\mathsf{T}\cdots \times_I \bm e_I^\mathsf{T}$, as the signal is real.  
% We first show the relation between $\hat{\mathcal{\bm Y}}$ and $\bar{\mathcal{\bm Y}}$. We note the $i^*$th mode matricization of $\hat{\mathcal{\bm Y}}$, i.e., $\hat{\bm Y}_{(i^*)}$, is \cite{cichocki2015tensor}
% \begin{equation}\label{eq.y_hat_mode}
%     \hat{\bm Y}_{(i^*)}=\bm e_{i^*} \xi \left(\left(\otimes_{i=I}^{i^*+1}\bm e_i \right) \otimes \left(\otimes_{i=i^*-1}^{1}\bm e_i \right)\right)^\mathsf{T}.
% \end{equation}
% Since
% \begin{multline}
%     \xi = \bar{\mathcal{\bm Y}} \times_1 \bm e_1^\mathsf{T}\cdots \times_I \bm e_I^\mathsf{T}\\
%     =\bm e_{i^*}^\mathsf{T} \bar{\bm Y}_{(i^*)} \left(\left(\otimes_{i=I}^{i^*+1}\bm e_i\bm e_i^\mathsf{T} \right) \otimes \left(\otimes_{i=i^*-1}^{1}\bm e_i\bm e_i^\mathsf{T} \right)\right)^\mathsf{T},
% \end{multline}
% which leads to
% \begin{equation}
%     \hat{\bm Y}_{(i^*)}=(\bm e_{i^*}\bm e_{i^*}^\mathsf{T})\bar{\bm Y}_{(i^*)}\left(\left(\otimes_{i=I}^{i^*+1}\bm e_i^\mathsf{T} \right) \otimes \left(\otimes_{i=i^*-1}^{1}\bm e_i^\mathsf{T} \right)\right)^\mathsf{T}.
% \end{equation}
% This holds for all $i^*\in [I]$. Hence, substituting $\xi$ into~\eqref{eq.y_hat_mode} leads to $\hat{\mathcal{\bm Y}}=\hat{\mathcal{\bm Y}}_{(I)}$ where $\hat{\mathcal{\bm Y}}_{(i)}=\bar{\mathcal{\bm Y}} \times_1 (\bm e_1 \bm e_1^\mathsf{T}) \times_2\cdots \times_i (\bm e_i \bm e_i^\mathsf{T})$. 
% This holds for all $i^*\in [I]$. 
Lemma~\ref{lm.mixedproduct} implies
\begin{equation*}
    \hat{\mathcal{\bm Y}}=\left(\mathcal{\bm Y} + \bar{\mathcal{N}} \right) \times_1 (\bm e_1 \bm e_1^\mathsf{T}) \times_2\cdots \times_I (\bm e_I \bm e_I^\mathsf{T}) = \bar{\mathcal{N}}_{(I)}+{\mathcal{\bm Y}}_{(I)},
\end{equation*}
where we define $\bar{\mathcal{N}}_{(I)}=\bar{\mathcal{N}}_{(I)} \times_1 (\bm e_1 \bm e_1^\mathsf{T}) \cdots \times_I (\bm e_I \bm e_I^\mathsf{T})$ and ${\mathcal{\bm Y}}_{(i)}={\mathcal{\bm Y}} \times_1 (\bm e_1 \bm e_1^\mathsf{T})\cdots \times_i (\bm e_i \bm e_i^\mathsf{T})$ with ${\mathcal{\bm Y}}_{(0)}=\mathcal{\bm Y}$.
% Since $\bar{\mathcal{\bm Y}} = \mathcal{\bm Y} + \bar{\mathcal{N}}$, {where $\bar{\mathcal{N}}$ is the noise tensor rearranged from $\bar{\bm n}$ in~\eqref{eq.Kro_generate_data} with variance $\sigma^2_\mathrm{t}$}, we derive
Therefore,  
\begin{equation}
 \|\otimes_{i=1}^I \hat{\bm y}_i-\otimes_{i=1}^I \bm y_i\|_2\!=\left\| \hat{\mathcal{\bm Y}} - \mathcal{\bm Y} \right\|_\mathrm{F}\! \leq \!  \left\| {\mathcal{\bm Y}}_{(I)} - \mathcal{\bm Y} \right\|_\mathrm{F}+\left\| \bar{\mathcal{N}}_{(I)} \right\|_\mathrm{F}.\label{eq.error_bnd_terms}
\end{equation}
% \begin{align}
%  \left\| \hat{\mathcal{\bm Y}} - \mathcal{\bm Y} \right\|_\mathrm{F} %&=\left\| \left(\mathcal{\bm Y} + \bar{\mathcal{N}} \right) \times_1 (\bm e_1 \bm e_1^\mathsf{T}) \times_2\cdots \times_I (\bm e_I \bm e_I^\mathsf{T}) - \mathcal{\bm Y} \right\|_\mathrm{F}\notag\\
%  &\leq \left\| \bar{\mathcal{N}} \times_1 (\bm e_1 \bm e_1^\mathsf{T}) \cdots \times_I (\bm e_I \bm e_I^\mathsf{T}) \right\|_\mathrm{F} 
%  \!+ \!\left\| {\mathcal{\bm Y}}_{(I)} \!- \mathcal{\bm Y} \right\|_\mathrm{F}\notag\\
%  &= \left\| \bar{\mathcal{N}} \times_1 \bm e_1^\mathsf{T}  \cdots \times_I \bm e_I^\mathsf{T}  \right\|_\mathrm{F} 
%  + \left\| {\mathcal{\bm Y}}_{(I)} - \mathcal{\bm Y} \right\|_\mathrm{F},\label{eq.error_bnd_terms}
% \end{align}
%since $\Vert\bm e_i\Vert_2=1$ and we define ${\mathcal{\bm Y}}_{(i)}={\mathcal{\bm Y}} \times_1 (\bm e_1 \bm e_1^\mathsf{T})\cdots \times_i (\bm e_i \bm e_i^\mathsf{T})$ with ${\mathcal{\bm Y}}_{(0)}=\mathcal{\bm Y}$.
To bound the first term in \eqref{eq.error_bnd_terms}, let $\bm P_{i_\perp}$ be the projection matrix orthogonal to $\bm e_i$ so that $\mathcal{\bm Y} = \mathcal{\bm Y}\times_1(\bm P_{1_\perp}+(\bm e_1 \bm e_1^\mathsf{T}))$, leading~to
\begin{align*}
    \mathcal{\bm Y} &= \mathcal{\bm Y}_{(0)}\!\times_1\!\bm P_{1_\perp} +{\mathcal{\bm Y}}_{(1)}= \mathcal{\bm Y}_{(0)}\!\times_1\!\bm P_{1_\perp}+{\mathcal{\bm Y}}_{(1)}\!\times_2\!\bm P_{2_\perp}+{\mathcal{\bm Y}}_{(2)}\\&= \sum_{i=1}^I{\mathcal{\bm Y}}_{(i-1)}\times_{i} \bm P_{i_\perp}+{\mathcal{\bm Y}}_{(I)},
\end{align*}
Therefore, using triangle inequality, we obtain
\begin{align}
\left\| {\mathcal{\bm Y}}_{(I)} - \mathcal{\bm Y} \right\|_\mathrm{F} &\leq \sum_{i=1}^I\left\| {\mathcal{\bm Y}}_{(i-1)}\times_{i} \bm P_{i_\perp}\right\|_\mathrm{F}
\leq \sum_{i=1}^I \left\| \mathcal{\bm Y} \times_i \bm P_{i_\perp} \right\|_\mathrm{F} 
    \notag\\&= \sum_{i=1}^I \left\| \bm P_{i_\perp} {\bm Y}_{(i)} \right\|_\mathrm{F}\leq  \sum_{i=1}^I2\|\bm N_{(i)}\|.\label{eq.term1bnd}
\end{align}
The last step follows from Lemma \ref{lm.pertubed_x}, as $ \bm P_{i_\perp}$ is the projection matrix orthogonal to $\bm e_i$ and the rank of ${\bm Y}_{(i)}$ is 1 from~\eqref{eq.matricization}. Also, Lemma \ref{lm.gaussian_spectral} with $x = \sqrt{2M_i}$ and $\bm Z = \sigma_\mathsf{t}^{-1}\bm N_{(i)}$ implies that with probability at least $1-2e^{-M_i}$
\begin{equation*}
      \|\sigma_\mathsf{t}^{-1}\bm N_{(i)}\|\!\leq\! \sqrt{M_i} + \sqrt{\bar{M}/M_i}+\sqrt{2M_i}\leq 3\sqrt{M_i} + \sqrt{\bar{M}/M_i}.
\end{equation*}
From \eqref{eq.term1bnd}, with probability exceeding $1-2\sum_{i=1}^Ie^{-M_i}$,
\begin{equation}
    \left\| {\mathcal{\bm Y}}_{(I)} - \mathcal{\bm Y} \right\|_\mathrm{F} \geq 2\sigma_\mathsf{t}\sum_{i=1}^I \left[3\sqrt{M_i} + \sqrt{\bar{M}/M_i}\right].\label{eq.probbound1}
\end{equation}

Next, we bound the second term in \eqref{eq.error_bnd_terms}. We note that $\|\sigma_\mathsf{t}^{-1}\bar{\mathcal{N}}_{(I)}\|_\mathrm{F}^2$ is an $1$-dimensional projection of a zero mean unit variance Gaussian tensor and follows $\chi^2_{1}(0)$~\cite[Supplement Sec. C.4]{zhang2019optimal}. Lemma \ref{lm.chi_square} states that for any $x>0$ $\sigma_\mathsf{t}^{-1}\left\| \bar{\mathcal{N}}_{(I)} \right\|_\mathrm{F} \leq \sqrt{1 + 2\sqrt{x} + 2x} \leq 1+2\sqrt{x}$
with probability exceeding $1-e^{-x}$.
%Note $(1+\sqrt{x})^2+x \leq (1+2\sqrt{x})^2$. Taking the square root on both sides of~\eqref{eq.noise_proj}, we have
% \begin{equation*}
%     \left\| \bar{\mathcal{N}} \times_1 \bm e_1^\mathsf{T} \cdots \times_I \bm e_I^\mathsf{T} \right\|_\mathrm{F} \geq \sigma_\mathsf{t}\left(1+2\sqrt{x}\right),
% \end{equation*}
% with probability less than $e^{-x}$. 
Setting $x=\max_{1\leq i\leq I}M_i$ and combining with \eqref{eq.error_bnd_terms} and \eqref{eq.probbound1} using the union bound yields
% \begin{multline*}
%     \left\| \hat{\mathcal{\bm Y}} - \mathcal{\bm Y} \right\|_\mathrm{F}\\\leq \sigma_\mathsf{t}\left(2\sum_{i=1}^I \left[3\sqrt{M_i} + \sqrt{\bar{M}/M_i}\right]+1 + 2\sqrt{\max_{1\leq i\leq I}M_i}\right),
% \end{multline*}
\eqref{eq.probabilistic} with probability exceeding $ 1-2\sum_{i=1}^Ie^{-M_i}-e^{-\max_{1\leq i\leq I}M_i} \geq 1-3\sum_{i=1}^Ie^{-M_i}$.
% \begin{equation*}
%     1-2\sum_{i=1}^Ie^{-M_i}-e^{-\max_{1\leq i\leq I}M_i} \geq 1-3\sum_{i=1}^Ie^{-M_i}.
% \end{equation*}
Hence, the proof is complete. 
% if
% \begin{multline}
%     \mathbb{P}\left( S \geq \sum_{i=1}^I 6\sqrt{M_i} + 2\left( \prod_{l=1,l\neq i}^I M_i \right)^{1/2} +1 + 2\sqrt{M} \right)\\
%     \leq \sum_{i=1}^I \mathbb{P}\left( 2\|\bm N_{(i)}\| \geq 6\sqrt{M_i} + \left(\prod_{l=1,l\neq i}^IM_i\right)^{1/2} \right)\\
%     +\mathbb{P}\left( \left\| \mathcal{N} \times_1 \bm P_{1} \cdots \times_I \bm P_{I} \right\|_\mathrm{F} \geq 1 + 2\sqrt{M} \right)\\
%     \leq 2 \sum_{i=1}^I e^{-M_i} + e^{-M} \leq 2Ie^{-cM} + e^{-M},
% \end{multline}
% indicating
% \begin{equation*}
% \left\| \hat{\mathcal{\bm Y}} - \mathcal{\bm Y} \right\|_\mathrm{F} \leq \left(6I\sqrt{C}+2\right)\sqrt{M} + 2\left( CM\right)^{(I-1)/2} +1
% \end{equation*}
% with probability at least $1 - 2Ie^{-cM} - e^{-M}$. We then have
% \begin{equation*}
%  \mathbb{P}\left(\left\| \hat{\mathcal{\bm Y}} - \mathcal{\bm Y} \right\|^2_\mathrm{F} \leq C_1M + C_2M^{I-1} \right) \geq 1 - C_3e^{-cM},
% \end{equation*}
% for positive constant $C_1$, $C_2$, and $C_3$ irrelevant to $\sigma_\mathrm{t}$ and $M_i$ for $i\in[I]$ but only $I$.

\bibliographystyle{IEEEtran}
\bibliography{refs}
\end{document}